\documentclass[12 pt,a4paper]{article} 

\usepackage{graphicx}
\usepackage{amsthm}
\usepackage{mathtools}
 \graphicspath{{fig/}}
 \usepackage{mathtools}
\newtheorem{proposition}{Proposition}

\title{On the Obtaining Solutions of Nonlinear Differential 
Equations by Means of the Solutions of Simpler Linear or Nonlinear Differential Equations}

\author{Nikolay  K. Vitanov}

\date{Institute of Mechanics, Bulgarian Academy of Sciences, Acad.
G. Bonchev Str., Bl. 4, 1113 Sofia, Bulgaria; vitanov@imbm.bas.bg
}

\begin{document}
\maketitle
\begin{abstract}
Transformations are much used to connect complicated nonlinear differential equations
 to simple equations with known exact solutions. Two examples of this are the
Hopf--Cole transformation and the simple equations method. In this article, we
follow an idea that is opposite to the idea of Hopf and Cole: we use transformations in
order to transform simpler linear or nonlinear differential equations (with known
solutions) to more complicated nonlinear differential equations. In such a way, we can 
obtain numerous exact solutions of nonlinear differential equations. We apply this
methodology to the classical parabolic differential equation (the wave equation), to
the classical hyperbolic differential equation (the heat equation), and to the classical
elliptic differential equation (Laplace equation). In addition, we use the methodology
to obtain exact solutions of nonlinear ordinary differential equations by means of
the solutions of linear differential equations and by means of the solutions of the
nonlinear differential equations of Bernoulli and Riccati. Finally, we demonstrate the
capacity of the methodology to lead to exact solutions of nonlinear partial 
differential equations on the basis of known solutions of other nonlinear partial
differential equations. As an example of this, we use the Korteweg--de Vries equation
and its solutions. Traveling wave solutions of nonlinear differential equations are
of special interest in this article. We demonstrate the existence of the following
phenomena described by some of the obtained solutions: (i) occurrence of the solitary wave--solitary antiwave
from the solution, which is zero at the initial moment (analogy of an occurrence of particle and
antiparticle from the vacuum); (ii) splitting of a nonlinear solitary wave into two solitary
waves (analogy of splitting of a particle into two particles); (iii) soliton behavior
of some of the obtained waves; (iv) existence of solitons which move with the
same velocity despite the different shape and amplitude of the solitons.
\end{abstract}

\section{Introduction}
Many complex systems can be modeled using nonlinear differential equations 
\cite{m1}-\cite{m7k}. The obtaining of
the exact solutions of these equations is important because the exact solutions allow us
to understand better the role of the parameters of the model for (i)~the evolution of
the studied phenomenon, and (ii) the regimes of the functioning of the studied
system. In addition, the exact solutions can be used for verification of the
computer programs designed to study complicated situations in the studied
systems.

There are various methods for obtaining exact solutions of nonlinear differential 
equations. Some of these methods such as the inverse scattering transform method or the
method of Hirota are very famous \cite{m8}-\cite{m14}. Such methodologies
often use transformations and lead to numerous exact solutions
of various nonlinear equations, and especially interesting are the localized solutions
such as solitons \cite{l1}-\cite{l5}. The methodology which is discussed below is
inspired by our work on the simple equations method (SEsM) for obtaining exact
solutions of nonlinear differential equations \cite{se1}-\cite{se9}.
An important step in the SEsM is the transformation, which has the goal to
remove the nonlinearity or to reduce the nonlinearity to a more tractable kind of
nonlinearity (such as, for example, polynomial nonlinearity). This step was inspired
by the successful attempt of Hopf and Cole \cite{hopf,cole}, who
managed to remove the nonlinearity of the B{\"u}rgers equation by means of appropriate
transformation and the nonlinear B{\"u}rgers equation was reduced to the linear heat 
equation.  
 
The idea of Hopf and Cole was to transform a nonlinear differential equation into a
linear differential equation with a known solution. The opposite idea is to transform
a linear differential equation with a known solution to a nonlinear differential 
equation. This idea will be followed below in the text. Thus, we will start from a 
linear differential equation with known solution. We will apply a transformation which
transforms the linear differential equation to a nonlinear differential equation. The
same transformation will transform the solution of the linear differential equation
to a solution of the nonlinear differential equation. 

 We start from a linear  differential equation $\hat{L}(u)=0$ (below, $\hat{L}$ and
 $\hat{N}$ denotes a linear or nonlinear differential operator) and perform
 a transformation $T$ in order to transform this equation into  a nonlinear
differential equation $\hat{N}(\phi)=0$
\begin{equation}\label{idea}
\hat{L} (u) =0 \xmapsto{T=u(\phi)} \hat{N} (\phi) =0.
\end{equation} 

Then, 
 we can use the solution of the linear differential equation in
order to obtain solutions of the nonlinear differential equation. We note that
different transformations $T$ can transform the linear equation into different
nonlinear equations. This will be the subject of the discussion below.
 
We can extend this idea as follows. Above, we start from a linear differential equation
with a known solution. We can also start  from a nonlinear differential equation with
a known solution. An appropriate transformation will transform this equation to
another nonlinear differential equation. Thus, we can obtain an exact solution to the
resulting nonlinear differential equation.
\begin{equation}\label{idea1}
\hat{N}_1 (u) =0 \xmapsto{T=u(\phi)} \hat{N}_2 (\phi) =0.
\end{equation}

In general, we assume that the
transformation is
\begin{equation}\label{transf}
T = u \left( \phi, \frac{\partial \phi}{\partial x}, \frac{\partial \phi}{\partial t}, \frac{\partial^2 \phi}{\partial x^2}, \frac{\partial^2 \phi}{\partial t^2}, \frac{\partial^2 \phi}{\partial x \partial t}, \dots \right).
\end{equation}

Below, we discuss several specific cases of this transformation.
For these special cases, we can easily calculate $\phi(u)$.
In such a way, we obtain exact analytical solutions of the
corresponding nonlinear differential equations.

The idea is extremely simple. Nevertheless, we show that it can lead to very
interesting results. The rest of the text is organized as follows. In Section \ref{sec2}, we
discuss several nonlinear equations and their solitons, which are obtained by 
transformation of the simple classical parabolic differential equation (the wave
equation), the classical hyperbolic  differential equation (the heat equation), and  the classical
elliptic differential equation (Laplace equation). We focus our attention on the wave
equation and show that the nonlinear equations which can be obtained on the basis
of the linear wave equation can describe interesting phenomena such as  (i) occurrence of
the solitary wave--solitary antiwave from solution, which is zero at the initial
moment (analogy of occurrence of particle and antiparticle from the vacuum); (ii) 
splitting of a nonlinear solitary wave into two solitary waves (analogy of splitting of a
particle into two particles); (iii) soliton behavior of some of the obtained waves;
(iv)~the existence of solitons which move at the same velocity despite different 
shapes and amplitudes of the solitons.
In Section \ref{sec3}, we discuss equations
and their solutions obtained by transformations from several simple linear ordinary 
differential equations, as well as by transformations of the differential equations of
Bernoulli and Riccati. In Section \ref{sec4}, we show that the application of the methodology to
nonlinear partial differential equations also leads to interesting results. We
apply transformations to the Kortwerg--de Vries equation and obtain other nonlinear
partial differential equations which possess multisoliton solutions. Several
concluding remarks are summarized in Section \ref{sec5}. Appendix \ref{appa} contains the solutions
of the linear partial differential equations which are used in the main text.
\section{Transformations of Linear Equations and Exact Solutions to the
Corresponding  Nonlinear Equations\label{sec2}}
\subsection{Transformations for the Wave Equation}
Let us consider the linear wave equation (\ref{wave}).
The solutions of this equation used in this text are presented in  Appendix \ref{appa}.
The simplest specific case of the transformation for $u$ is
\begin{equation}\label{transf1}
u =u(\phi).
\end{equation}

The wave equation becomes
\begin{equation}\label{twe1}
 \frac{du}{d \phi} \frac{\partial^2 \phi}{\partial t^2}  - c^2 \frac{d u}{d \phi}  \frac{\partial^2 \phi}{\partial x^2} +
\frac{d^2 u }{d \phi^2}  \left(\frac{\partial \phi}{\partial t} \right)^2 - c^2 \frac{d^2 u }{d \phi^2} \left(\frac{\partial \phi}{\partial x} \right)^2
 = 0.
\end{equation}

In order to specify (\ref{twe1}), we have to specify the form of the
transformation (\ref{transf1}). 
\begin{proposition}
The equation
\begin{equation}\label{eq1}
 \frac{\partial^2 \phi}{\partial t^2} - c^2 \frac{\partial^2 \phi}{\partial x^2} + \left(\frac{\partial \phi}{\partial t}\right)^2  - c^2 \left(\frac{\partial \phi}{\partial x} \right)^2  =0,
\end{equation}
has the solutions 
\begin{eqnarray}\label{sol1_eq1}
\phi(x,t) =  \ln \Bigg \{ F(x+ct) + G(x-ct)\Bigg \},
\end{eqnarray}
where $F$ and $G$ are arbitrary $C^2$-functions, and 
$-\infty \le a <x < b \le + \infty$, and $t>0$. Another solution is
\begin{eqnarray}\label{sol2_eq1}
\phi(x,t) = \ln \Bigg \{ 
\frac{f(x+ct) + f(x-ct)}{2} + \frac{1}{2c} \int \limits_{x-ct}^{x+ct} ds \ g(s)
\Bigg \}, 
\end{eqnarray}
for the conditions 
\begin{eqnarray}\label{sol2_cauchy}
\phi(x,0)= \ln[f(x)], \ \ \ \frac{\partial \exp(\phi) }{\partial t}(x,0) = g(x), \ \ 
-\infty < x < \infty, t>0,
\end{eqnarray}
and the solution
\begin{eqnarray}\label{sol3_eq1}
\phi(x,t) = \ln \Bigg \{ \sum \limits_{n=1}^\infty \Bigg[a_n \cos \Bigg( \frac{n\pi ct}{L} \Bigg)+b_n \sin \Bigg(\frac{n\pi ct}{L} \Bigg) \Bigg] \sin \Bigg( \frac{n\pi x}{L}\Bigg) \Bigg \}, \nonumber \\
a_n = \frac{2}{L}\int \limits_0^L dx \  f(x) \sin \Bigg( \frac{n \pi x}{L} \Bigg), \ \ \ b_n = \frac{2}{n \pi c} \int \limits_0^L dx \ 
g(x) \sin \Bigg( \frac{n \pi x}{L} \Bigg), 
\end{eqnarray}
for the case $ 0<x<L$ and initial and
boundary conditions 
\begin{eqnarray}\label{wave_fourier}
\phi(x,0) = \ln[f(x)], \ \ \ \
\frac{\partial \exp(\phi)}{\partial t}(x,0)=g(x), \ \ 
0 \le x \le L \nonumber \\
\frac{\partial \exp(\phi)}{\partial x}(0,t) = \frac{\partial \exp(\phi)}{\partial x}(L,t) = 0, \ \ t \ge 0.
\end{eqnarray}
\end{proposition}
\begin{proof}
Let us consider the linear wave equation (\ref{wave}).
We performed the transformation
\begin{equation}\label{t1}
u = \exp(\phi).
\end{equation}

The transformation (\ref{t1}) transforms Equation (\ref{wave}) to Equation (\ref{eq1}). Equation (\ref{wave}) has the solutions  (\ref{wavesol}), (\ref{wavesol2}) and (\ref{wavesol3}). Thus,
Equation (\ref{eq1}) has the solutions (\ref{sol1_eq1}), (\ref{sol2_eq1}) and (\ref{sol3_eq1}).
\end{proof}
We note that a $C^2$-function is a function which has two continuous derivatives.
\begin{proposition}
The equation
\begin{equation}\label{eq1_2}
\phi \frac{\partial^2 \phi}{\partial t^2} - c^2 \phi \frac{\partial^2 \phi}{\partial x^2} +
(\alpha-1)\left(\frac{\partial \phi}{\partial t}\right)^2  - c^2(\alpha-1) \left(\frac{\partial \phi}{\partial x} \right)^2 =0,
\end{equation}
has the solutions 
\begin{eqnarray}\label{sol1_eq12}
\phi(x,t) =   \Bigg \{ F(x+ct) + G(x-ct)\Bigg \}^{1/\alpha},
\end{eqnarray}
where $F$ and $G$ are arbitrary $C^2$-functions, and 
$-\infty \le a <x < b \le + \infty$, and $t>0$. Another solution is
\begin{eqnarray}\label{sol2_eq12}
\phi(x,t) =  \Bigg \{ 
\frac{f(x+ct) + f(x-ct)}{2} + \frac{1}{2c} \int \limits_{x-ct}^{x+ct} ds \ g(s)
\Bigg \}^{1/\alpha}, 
\end{eqnarray}
for the conditions 
\begin{eqnarray}\label{sol2_cauchyxx}
\phi(x,0)= [f(x)]^{1/\alpha}, \ \ \ \frac{\partial (\phi^\alpha) }{\partial t}(x,0) = g(x), \ \ 
-\infty < x < \infty, t>0
\end{eqnarray}
and the solution
\begin{eqnarray}\label{sol3_eq12}
\phi(x,t) =  \Bigg \{ \sum \limits_{n=1}^\infty \Bigg[a_n \cos \Bigg( \frac{n\pi ct}{L} \Bigg)+b_n \sin \Bigg(\frac{n\pi ct}{L} \Bigg) \Bigg] \sin \Bigg( \frac{n\pi x}{L}\Bigg) \Bigg \}^{1/\alpha}, \nonumber \\
a_n = \frac{2}{L}\int \limits_0^L dx \  f(x) \sin \Bigg( \frac{n \pi x}{L} \Bigg), \ \ \ b_n = \frac{2}{n \pi c} \int \limits_0^L dx \ 
g(x) \sin \Bigg( \frac{n \pi x}{L} \Bigg), 
\end{eqnarray}
for the case $ 0<x<L$ and initial and
boundary conditions 
\begin{eqnarray}\label{wave_fourier2}
\phi(x,0) = [f(x)]^{1/\alpha}, \ \ \ \
\frac{\partial (\phi^\alpha)}{\partial t}(x,0)=g(x), \ \ 
0 \le x \le L \nonumber \\
\frac{\partial (\phi^\alpha)}{\partial x}(0,t) = \frac{\partial (\phi^\alpha)}{\partial x}(L,t) = 0, \ \ t \ge 0.
\end{eqnarray}
\end{proposition}
\begin{proof}
Let us consider the linear wave Equation (\ref{wave}).
We performed the transformation
\begin{equation}\label{t1_2}
u = \phi^\alpha, \alpha \ne 1.
\end{equation}

The transformation (\ref{t1_2}) transforms Equation (\ref{wave}) to
 Equation (\ref{eq1_2}).
 Equation~(\ref{wave}) has the solutions  (\ref{wavesol}), (\ref{wavesol2}) and (\ref{wavesol3}). Thus, 
Equation (\ref{eq1_2}) has the solutions (\ref{sol1_eq12}),  (\ref{sol2_eq12}) and (\ref{sol3_eq12}).
\end{proof}

An example of solution (\ref{sol1_eq12}) is presented in Figure \ref{fig1}. We observe
that the transformation can change the form and the orientation of the solution.
\begin{figure}[h]
\centering
\includegraphics[scale=0.45,angle=-90]{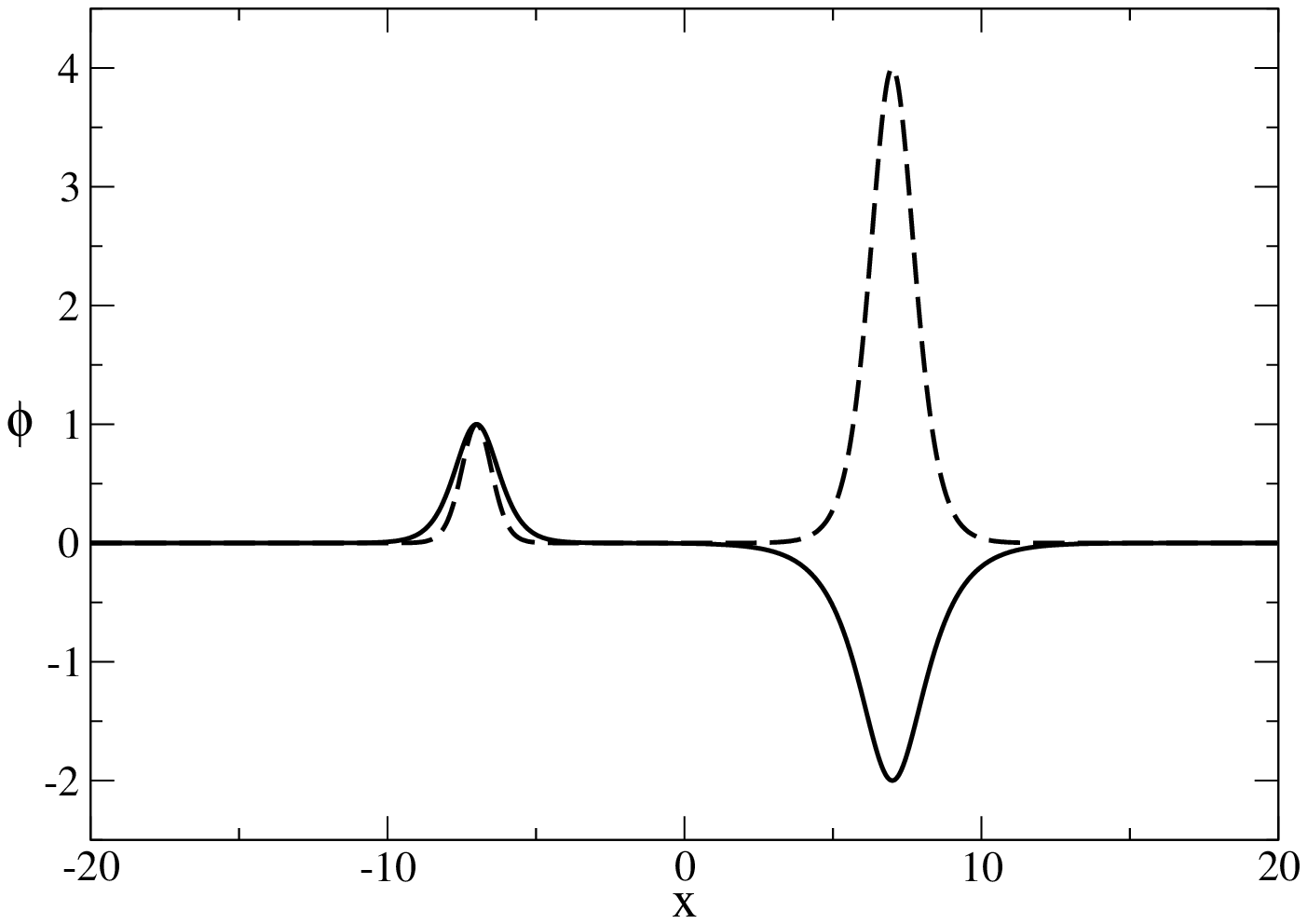}
\caption{Illustration 
 of solution (\ref{sol1_eq12}) of Equation (\ref{eq1_2}). Solid line: corresponding solution of the linear wave equation. 
The parameters of the solution are $F=1/\cosh^2(x+ct)$;
\mbox{$G=-2/\cosh(x-ct)$;} $c=2.0$; $t=3.5$.
Dashed line: solution (\ref{sol1_eq12}). $\alpha=1/2$. \label{fig1}}
\end{figure}

Figures \ref{fig2} and \ref{fig3} illustrate two interesting phenomena connected to some of the obtained
solutions to the nonlinear differential equation. Figure \ref{fig2} illustrates the occurrence of
a wave and an antiwave from  $\phi=0$ at $t=0$. This phenomenon is similar to the
arising of a particle and antiparticle from the vacuum. Note that despite the fact that $\phi=0$ at $t=0$, the corresponding solution of (\ref{sol1_eq12}) has an internal structure.

\begin{figure}[h]
\centering
\includegraphics[scale=0.45,angle=-90]{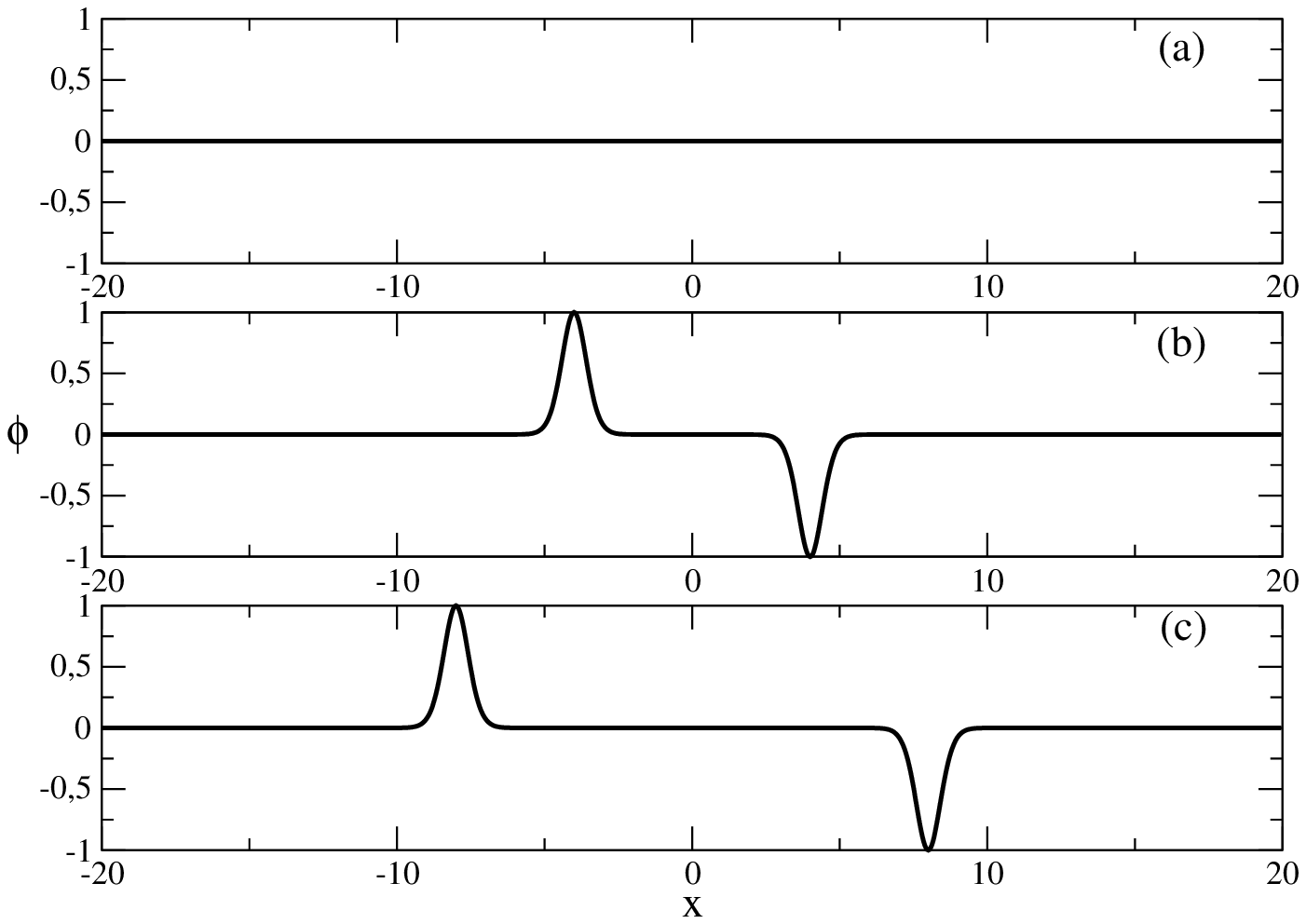}
\caption{Illustration 
 of the wave--antiwave phenomenon connected to solution (\ref{sol1_eq12}) to the nonlinear partial differential Equation (\ref{eq1_2}).
The parameters of the solution are $F=1/\cosh^2(x+ct)$;
$G = -1/\cosh^2(x-ct)$; $c=2.0$; $\alpha=1/3$. (\textbf{a}) $t=0$. (\textbf{b}) $t=2$. (\textbf{c}) $t=4$. \label{fig2}
}
\end{figure}

Figure \ref{fig3} illustrates the phenomenon of a splitting of a nonlinear wave. This
phenomenon is similar to the phenomenon of the splitting of a particle into two other particles.
Note the nonlinearity of the superposition of the two waves in Figure \ref{fig3}.\vspace{-3pt}
\begin{figure}[h]
\centering
\includegraphics[scale=0.45,angle=-90]{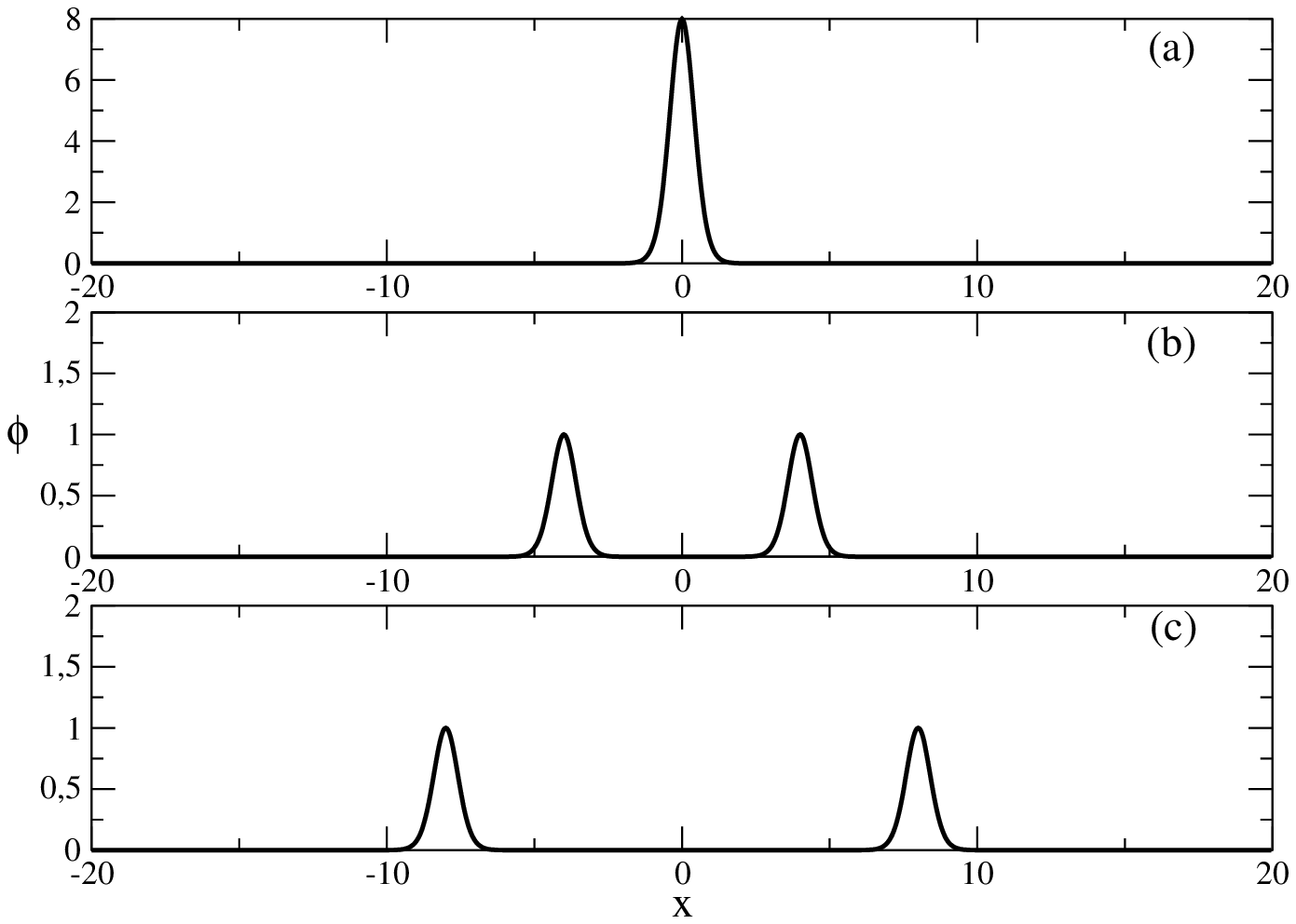}
\caption{Illustration 
 of the wave splitting phenomenon connected to solution
(\ref{sol1_eq12}) to the nonlinear partial differential Equation (\ref{eq1_2}).
The parameters of the solution are $F=1/\cosh^2(x+ct)$;\linebreak  
$G=-1/\cosh^2(x-ct)$; $c=2.0$; $\alpha=1/3$. (\textbf{a}) $t=0$. (\textbf{b}) $t=2$. (\textbf{c}) $t=4$. \label{fig3}
}
\end{figure}


%
\begin{proposition}
The equation
\begin{equation}\label{eq1_3}
- \phi \frac{\partial^2 \phi}{\partial t^2}+ c^2 \phi \frac{\partial^2 \phi}{\partial x^2}+
\left(\frac{\partial \phi}{\partial t}\right)^2 - c^2 \left(\frac{\partial \phi}{\partial x} \right)^2  =0,
\end{equation}
has the solutions 
\begin{eqnarray}\label{sol1_eq13}
\phi(x,t) =  \exp \Bigg \{ F(x+ct) + G(x-ct)\Bigg \},
\end{eqnarray}
where $F$ and $G$ are arbitrary $C^2$-functions, and 
$-\infty \le a <x < b \le + \infty$, and $t>0$. Another solution is
\begin{eqnarray}\label{sol2_eq13}
\phi(x,t) =  \exp \Bigg \{ 
\frac{f(x+ct) + f(x-ct)}{2} + \frac{1}{2c} \int \limits_{x-ct}^{x+ct} ds \ g(s)
\Bigg \}, 
\end{eqnarray}
for the conditions 
\begin{eqnarray}\label{sol2_cauchy3}
\phi(x,0)= \exp [f(x)], \ \ \ \frac{\partial \ln (\phi) }{\partial t}(x,0) = g(x), \ \ 
-\infty < x < \infty, t>0,
\end{eqnarray}
and the solution
\begin{eqnarray}\label{sol3_eq13}
\phi(x,t) =  \exp \Bigg \{ \sum \limits_{n=1}^\infty \Bigg[a_n \cos \Bigg( \frac{n\pi ct}{L} \Bigg)+b_n \sin \Bigg(\frac{n\pi ct}{L} \Bigg) \Bigg] \sin \Bigg( \frac{n\pi x}{L}\Bigg) \Bigg \}, \nonumber \\
a_n = \frac{2}{L}\int \limits_0^L dx \  f(x) \sin \Bigg( \frac{n \pi x}{L} \Bigg), \ \ \ b_n = \frac{2}{n \pi c} \int \limits_0^L dx \ 
g(x) \sin \Bigg( \frac{n \pi x}{L} \Bigg), 
\end{eqnarray}
for the case $ 0<x<L$ and initial and
boundary conditions 
\begin{eqnarray}\label{wave_fourier3}
\phi(x,0) = \exp [f(x)], \ \ \ \
\frac{\partial \ln (\phi)}{\partial t}(x,0)=g(x), \ \ 
0 \le x \le L \nonumber \\
\frac{\partial \ln(\phi)}{\partial x}(0,t) = \frac{\partial ln (\phi)}{\partial x}(L,t) = 0, \ \ t \ge 0.
\end{eqnarray}
\end{proposition}
\begin{proof}
Let us consider the linear wave Equation (\ref{wave}).
We performed the transformation
\begin{equation}\label{t1_3}
u = \ln(\phi).
\end{equation}

The transformation (\ref{t1_2}) transforms Equation (\ref{wave}) to Equation (\ref{eq1_3}).
Equation~(\ref{wave}) has the solutions  (\ref{wavesol}), (\ref{wavesol2}) and (\ref{wavesol3}). Thus, 
Equation (\ref{eq1_2}) has the solutions (\ref{sol1_eq13}),  (\ref{sol2_eq13}) and (\ref{sol3_eq13}).
\end{proof}

Figure \ref{fig4} illustrates a solution of (\ref{eq1_3}). Note that the transformation
transforms the solution of the wave equation to a solution in which waves have
non-negative values.\vspace{-6pt}
\begin{figure}[h]
\centering
\includegraphics[scale=0.45,angle=-90]{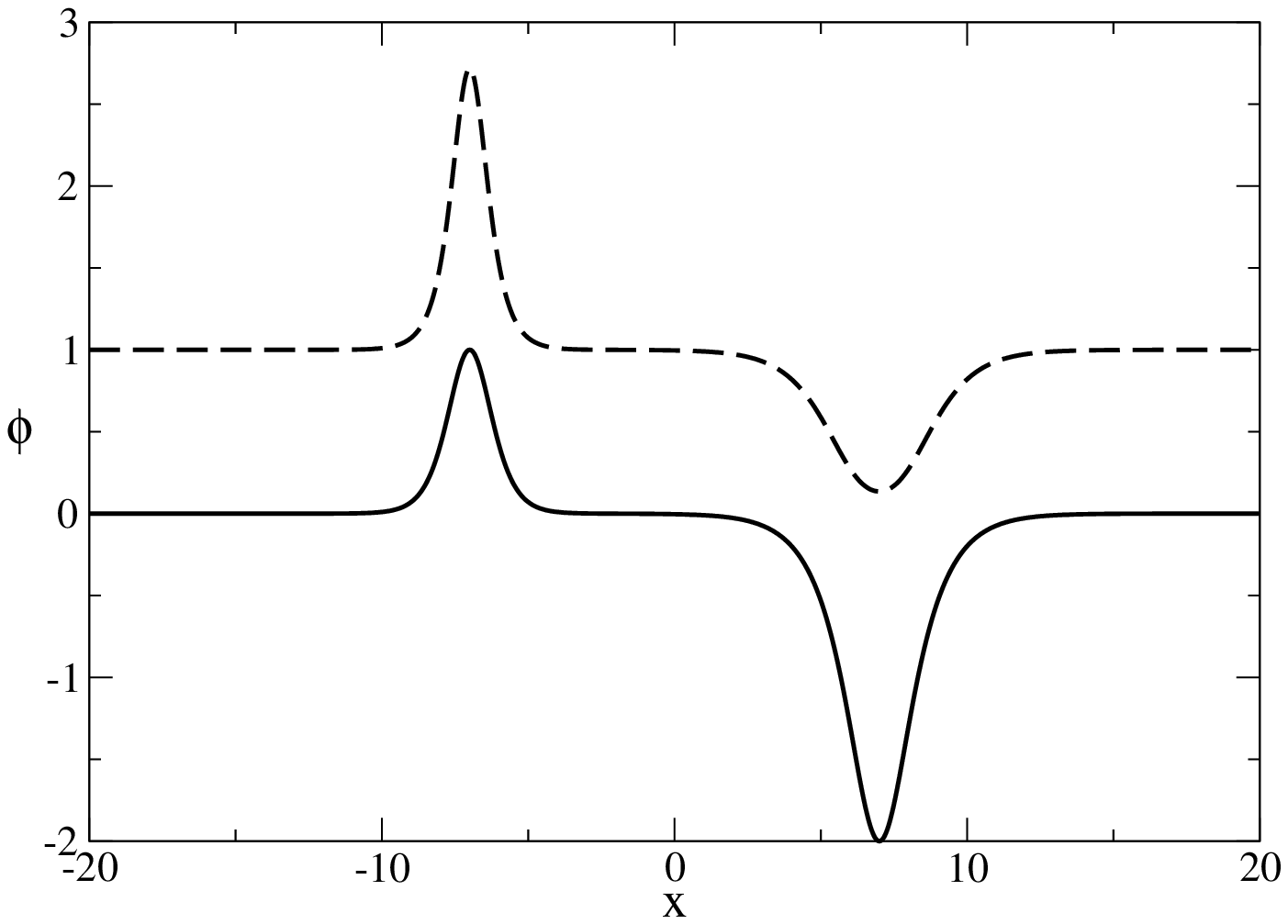}
\caption{Illustration 
 of solution (\ref{sol1_eq13}) of Equation (\ref{eq1_3}). Solid line: corresponding solution of the linear wave equation. 
The parameters of the solution are $F=1/\cosh^2(x+ct)$;
\mbox{$G=-2/\cosh(x-ct)$;} $c=2.0$; $t=3.5$.
Dashed line: solution (\ref{sol1_eq13}).\label{fig4}}
\end{figure}


Figure \ref{fig5} illustrates the typical characteristics of soliton behavior for some of the
obtained solutions. Figure \ref{fig5}a--c present a collision of two solitary waves
which are described by a solution of Equation (\ref{eq1_3}). We see that the
form of the solitary waves after the collision is the same as their form before the collision. Figure \ref{fig5}b presents an interesting moment of the collision which is dominated by the wave of smaller amplitude. We note that for the classical solitons,
the amplitude of the soliton is a function of its velocity. Figure \ref{fig5} shows another
kind of soliton. These solitons have different amplitudes but nevertheless travel
with the same velocity. This property is quite interesting.

\begin{figure}[h]
\centering
\includegraphics[scale=0.45,angle=-90]{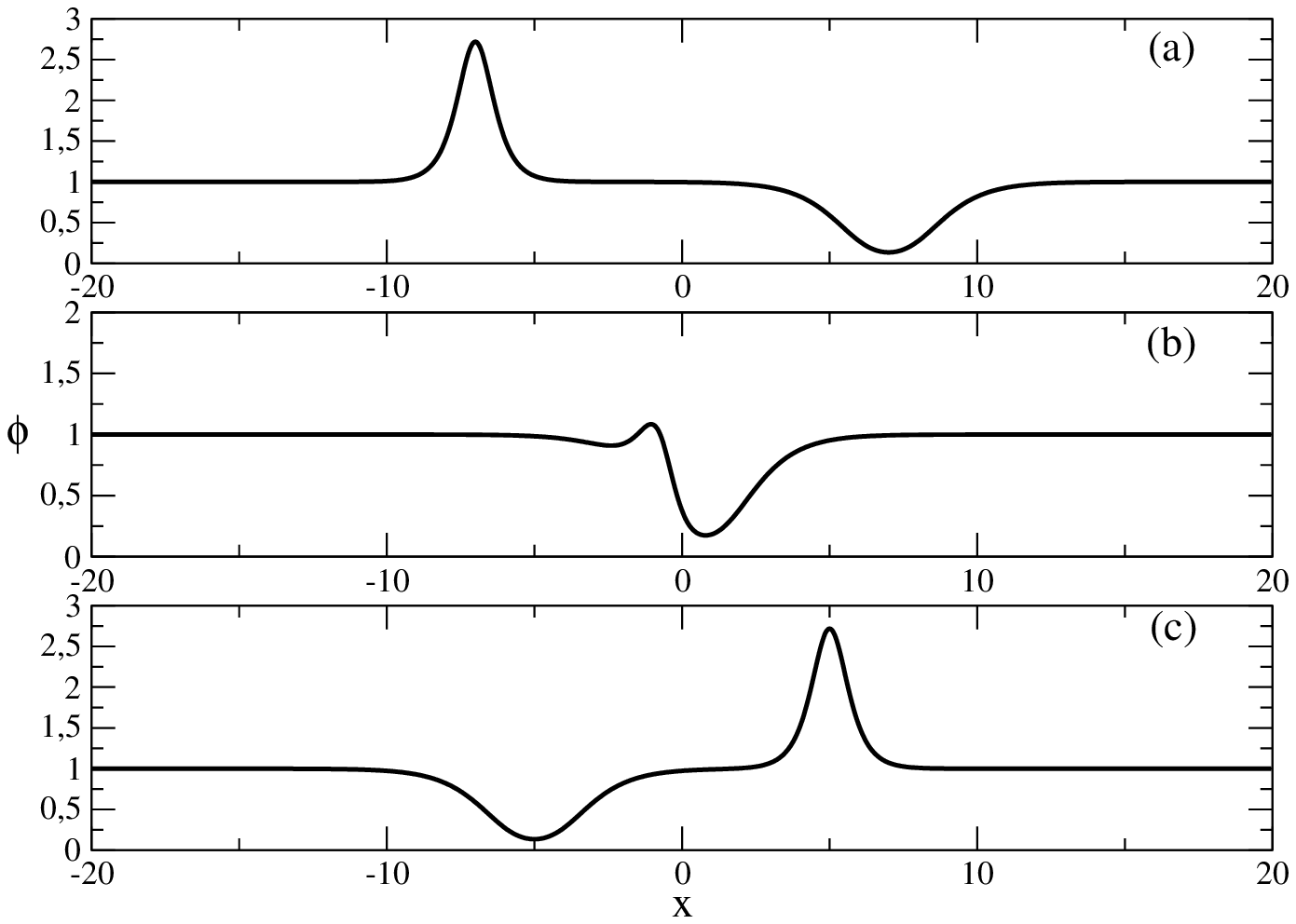}
\caption{Illustration 
 of the soliton properties of the waves connected to solution (\ref{sol1_eq13}) of Equation (\ref{eq1_3}).  
The parameters of the solution are $F=1/\cosh^2(x+ct)$;
$G=-2/\cosh(x-ct)$; $c=2.0$. (\textbf{a}) $t=-2.5$. (\textbf{b}) $t=0.3$. (\textbf{c}) $t=3.5$.\label{fig5}}
\end{figure}

The soliton properties of solution (\ref{sol1_eq13}) to Equation (\ref{eq1_3}) are also illustrated in Figure \ref{fig6}.
We observe the collision of the two solitons. The solitons merge into a single solitary wave and then split
again. The forms of the solitons do not change. Note, that the velocities of these solitons are the same despite the
differences in their profiles.

\begin{figure}[htb!]
\centering 
\includegraphics[scale=0.9]{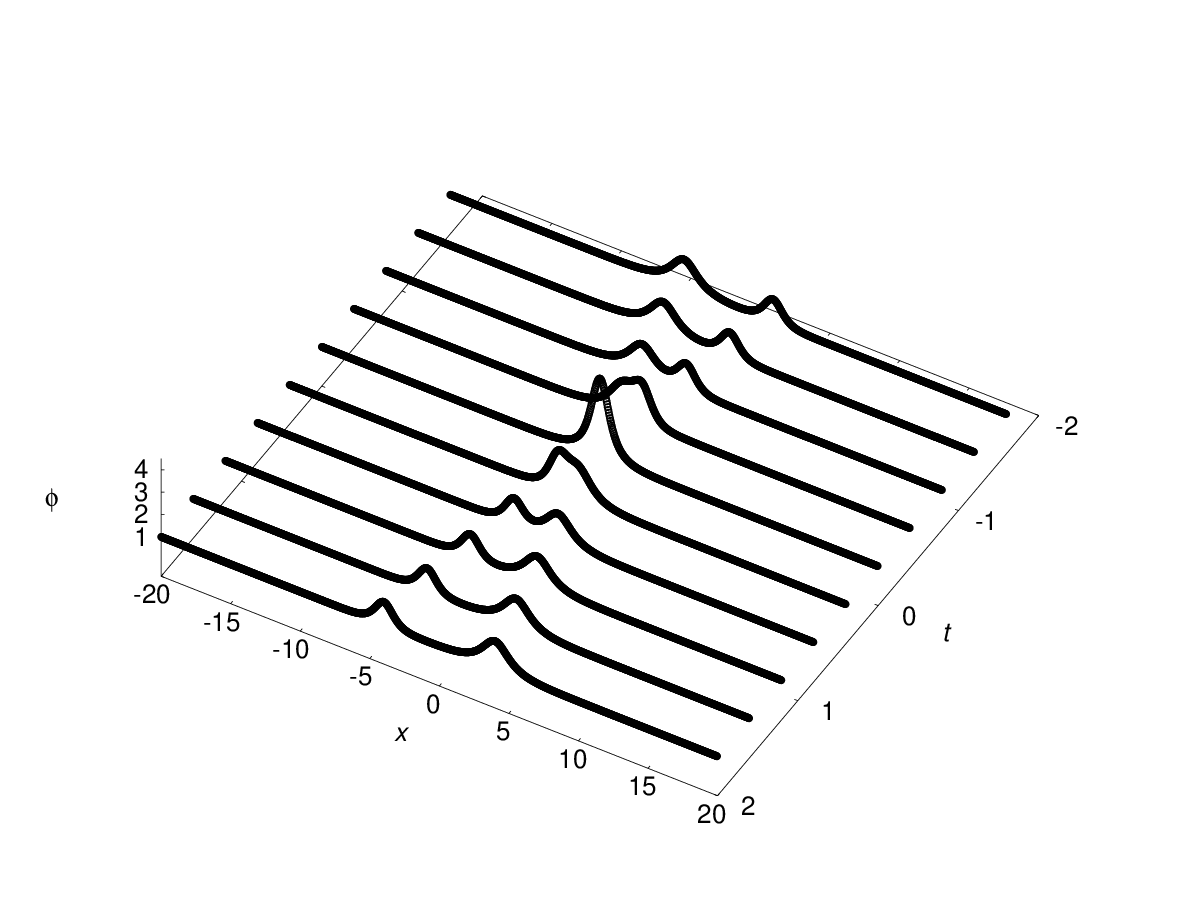}
\caption{Illustration of the soliton properties of the waves connected to solution (\ref{sol1_eq13}) to Equation (\ref{eq1_3}).
The parameters of the solution are $F=0.7/\cosh^2(x+ct)$;
$G=0.8/\cosh(x-ct)$; $c=2.0$.\label{fig6}}
\end{figure}

Let us now consider a more complicated transformation: $u=u(\phi, \phi_x)$.
\begin{proposition}
The equation
\begin{equation}\label{eq1_4}
c^2 \phi \frac{\partial^3 \phi}{\partial x^3} - \phi \frac{\partial^3 \phi}{\partial x \partial t^2} + 3c^2 \frac{\partial \phi}{\partial x}
\frac{\partial^2 \phi}{\partial x^2} - \frac{\partial \phi}{\partial x}
\frac{\partial \phi}{\partial t^2} - 2 \frac{\partial \phi}{\partial t}
\frac{\partial^2 \phi}{\partial x \partial t} =0,
\end{equation}
has the solutions 
\begin{eqnarray}\label{sol1_eq14}
\phi(x,t) =  2^{1/2} \Bigg \{ \int dx [F(x+ct) + G(x-ct)] \Bigg \}^{1/2},
\end{eqnarray}
where $F$ and $G$ are arbitrary $C^2$-functions, and 
$-\infty \le a <x < b \le + \infty$, and $t>0$. Another solution is
\begin{eqnarray}\label{sol2_eq14}
\phi(x,t) =  2^{1/2} \Bigg \{ \int dx \Bigg[ 
\frac{f(x+ct) + f(x-ct)}{2} + \frac{1}{2c} \int \limits_{x-ct}^{x+ct} ds \ g(s)
\Bigg ] \Bigg \}^{1/2}, 
\end{eqnarray}
for the conditions 
\begin{eqnarray}\label{sol2_cauchy4}
\phi(x,0) \frac{\partial \phi}{\partial x}(x,0)= f(x), \ \ \ \frac{\partial }{\partial t}\left[\phi \frac{\partial \phi}{\partial x} \right](x,0) = g(x), \ \ 
-\infty < x < \infty, t>0,
\end{eqnarray}
and the solution
\begin{eqnarray}\label{sol3_eq14}
\phi(x,t) =  2^{1/2} \Bigg \{ \int dx \Bigg[  \sum \limits_{n=1}^\infty \Bigg[a_n \cos \Bigg( \frac{n\pi ct}{L} \Bigg)+b_n \sin \Bigg(\frac{n\pi ct}{L} \Bigg) \Bigg] \sin \Bigg( \frac{n\pi x}{L}\Bigg) \Bigg ] \Bigg \}^{1/2}, \nonumber \\
a_n = \frac{2}{L}\int \limits_0^L dx \  f(x) \sin \Bigg( \frac{n \pi x}{L} \Bigg), \ \ \ b_n = \frac{2}{n \pi c} \int \limits_0^L dx \ 
g(x) \sin \Bigg( \frac{n \pi x}{L} \Bigg), 
\end{eqnarray}
for the case $ 0<x<L$ and initial and
boundary conditions 
\begin{eqnarray}\label{wave_fourier4}
\phi(x,0) = 2^{1/2} [\int dx  \ f(x)]^{1/2}, \ \ \ \
\frac{\partial }{\partial t}\left[\phi \frac{\partial \phi}{\partial x} \right](x,0) = g(x), \ \ 
0 \le x \le L \nonumber \\
\frac{\partial }{\partial x}\left[\phi \frac{\partial \phi}{\partial x} \right](0,t) = \frac{\partial }{\partial x}\left[\phi \frac{\partial \phi}{\partial x} \right](L,t) = 0, \ \ t \ge 0.
\end{eqnarray}
\end{proposition}
\begin{proof}
Let us consider the linear wave Equation (\ref{wave}).
We performed the transformation
\begin{equation}\label{t1_4}
u = \phi \frac{\partial \phi}{\partial x}.
\end{equation}

The transformation (\ref{t1_4}) transforms Equation (\ref{wave}) to Equation~(\ref{eq1_4}).
Equation~(\ref{lhe}) has the solutions  (\ref{wavesol}), (\ref{wavesol2}) and (\ref{wavesol3}). Thus,
Equation (\ref{eq1_4}) has the solutions (\ref{sol1_eq14}),  (\ref{sol2_eq14}) and (\ref{sol3_eq14}).
\end{proof}

Note that (\ref{t1_4}) leads to $\phi = 2^{1/2} \left( \int dx \ u \right)^{1/2}$.
\begin{proposition}
The equation
\begin{equation}\label{eq1_5}
(c^2 - 1)  \frac{\partial^3 \phi}{\partial x^3} + 3 c^2 \frac{\partial \phi}{\partial 
x} \frac{\partial^2 \phi}{\partial x^2} - 3  \frac{\partial \phi}{\partial t} \frac{\partial^2 \phi}{\partial t^2}+
c^2 \left( \frac{\partial \phi}{\partial x} \right)^3   - \left( \frac{\partial \phi}{\partial t} \right)^3  = 0,
\end{equation}
has the solutions 
\begin{eqnarray}\label{sol1_eq15}
\phi(x,t) =  \ln \Bigg \{ \int dx [F(x+ct) + G(x-ct)] \Bigg \},
\end{eqnarray}
where $F$ and $G$ are arbitrary $C^2$-functions, and 
$-\infty \le a <x < b \le + \infty$, and $t>0$. Another solution is
\begin{eqnarray}\label{sol2_eq15}
\phi(x,t) =  \ln \Bigg \{ \int dx \Bigg[ 
\frac{f(x+ct) + f(x-ct)}{2} + \frac{1}{2c} \int \limits_{x-ct}^{x+ct} ds \ g(s)
\Bigg ] \Bigg \}, 
\end{eqnarray}
for the conditions 
\begin{eqnarray}\label{sol2_cauchy5}
\exp[\phi(x,0)] \frac{\partial \phi}{\partial x}(x,0)= f(x), \ \ \ \frac{\partial }{\partial t}\left[\exp(\phi) \frac{\partial \phi}{\partial x} \right](x,0) = g(x), \ \ 
-\infty < x < \infty, t>0,
\end{eqnarray}
and the solution
\begin{eqnarray}\label{sol3_eq15}
\phi(x,t) =  \ln \Bigg \{ \int dx \Bigg[  \sum \limits_{n=1}^\infty \Bigg[a_n \cos \Bigg( \frac{n\pi ct}{L} \Bigg)+b_n \sin \Bigg(\frac{n\pi ct}{L} \Bigg) \Bigg] \sin \Bigg( \frac{n\pi x}{L}\Bigg) \Bigg ] \Bigg \}, \nonumber \\
a_n = \frac{2}{L}\int \limits_0^L dx \  f(x) \sin \Bigg( \frac{n \pi x}{L} \Bigg), \ \ \ b_n = \frac{2}{n \pi c} \int \limits_0^L dx \ 
g(x) \sin \Bigg( \frac{n \pi x}{L} \Bigg), 
\end{eqnarray}
for the case $ 0<x<L$ and initial and
boundary conditions 
\begin{eqnarray}\label{wave_fourier5}
\phi(x,0) = \ln [\int dx  \ f(x)], \ \ \ \
\frac{\partial }{\partial t}\left[\phi \frac{\partial \phi}{\partial x} \right](x,0) = g(x), \ \ 
0 \le x \le L \nonumber \\
\frac{\partial }{\partial x}\left[\phi \frac{\partial \phi}{\partial x} \right](0,t) = \frac{\partial }{\partial x}\left[\phi \frac{\partial \phi}{\partial x} \right](L,t) = 0, \ \ t \ge 0.
\end{eqnarray}
\end{proposition}
\begin{proof}
Let us consider the linear wave Equation (\ref{wave}).
We performed the transformation
\begin{equation}\label{t1_5}
u = \exp(\phi)\  \frac{\partial \phi}{\partial x}.
\end{equation}

The transformation (\ref{t1_5}) transforms Equation (\ref{wave}) to
 Equation (\ref{eq1_5}).
Equation~(\ref{lhe}) has the solutions  (\ref{wavesol}), (\ref{wavesol2}) and (\ref{wavesol3}). Thus, 
Equation (\ref{eq1_5}) has the solutions (\ref{sol1_eq15}),  (\ref{sol2_eq15}) and (\ref{sol3_eq15}).
\end{proof}

Note that (\ref{t1_5}) leads to $\phi = \ln \left( \int dx \ u \right)$.
\begin{proposition}
The equation
\begin{eqnarray}\label{eq1_6}
2 \phi \left[ c^2 \left( \frac{\partial^2 \phi}{\partial x^2} \right)^2  - \left( \frac{\partial^2 \phi}{\partial x \partial t} \right)^2 \right] + \left( \frac{\partial \phi}{\partial x}  \right)^2 \left[ c^2 \frac{\partial^2 \phi}{\partial x^2} - \frac{\partial^2 \phi}{\partial t^2}  \right] - \nonumber \\
2 \frac{\partial \phi}{\partial x} \left( c^2 \frac{\partial \phi}{\partial x} \frac{\partial^2 \phi}{\partial x^2} - \frac{\partial \phi}{\partial t} \frac{\partial^2 \phi}{\partial x \partial t}\right)  - \phi \frac{\partial \phi}{\partial x}
\left(c^2 \frac{\partial^3 \phi}{\partial x^3} - \frac{\partial \phi }{\partial x \partial t^2} \right) =0,
\end{eqnarray}
has the solutions 
\begin{eqnarray}\label{sol1_eq16}
\phi(x,t) =  \exp \Bigg \{ \int dx \frac{1}{[F(x+ct) + G(x-ct)]} \Bigg \},
\end{eqnarray}
where $F$ and $G$ are arbitrary $C^2$-functions, and 
$-\infty \le a <x < b \le + \infty$, and $t>0$. Another solution is
\begin{eqnarray}\label{sol2_eq16}
\phi(x,t) =  \exp \Bigg \{ \int dx \frac{1}{\Bigg[ 
\frac{f(x+ct) + f(x-ct)}{2} + \frac{1}{2c} \int \limits_{x-ct}^{x+ct} ds \ g(s)
\Bigg ]} \Bigg \}, 
\end{eqnarray}
for the conditions 
\begin{eqnarray}\label{sol2_cauchy5xx}
\phi(x,0)] \frac{\partial \phi}{\partial x}(x,0)= f(x), \ \ \ \frac{\partial }{\partial t}\left[\phi \frac{\partial \phi}{\partial x} \right](x,0) = g(x), \ \ 
-\infty < x < \infty, t>0,
\end{eqnarray}
and the solution
\begin{eqnarray}\label{sol3_eq16}
\phi(x,t) =  \exp \Bigg \{ \int dx \frac{1}{ \Bigg[  \sum \limits_{n=1}^\infty \Bigg[a_n \cos \Bigg( \frac{n\pi ct}{L} \Bigg)+b_n \sin \Bigg(\frac{n\pi ct}{L} \Bigg) \Bigg] \sin \Bigg( \frac{n\pi x}{L}\Bigg) \Bigg ] }\Bigg \}, \nonumber \\
a_n = \frac{2}{L}\int \limits_0^L dx \  f(x) \sin \Bigg( \frac{n \pi x}{L} \Bigg), \ \ \ b_n = \frac{2}{n \pi c} \int \limits_0^L dx \ 
g(x) \sin \Bigg( \frac{n \pi x}{L} \Bigg), 
\end{eqnarray}
for the case $ 0<x<L$ and initial and
boundary conditions 
\begin{eqnarray}\label{wave_fourier5xx}
\phi(x,0) = \exp \left [\int dx  \frac{1}{f(x)} \right], \ \ \ \
\frac{\partial }{\partial t}\left[\phi \Bigg / \frac{\partial \phi}{\partial x} \right](x,0) = g(x), \ \ 
0 \le x \le L \nonumber \\
\frac{\partial }{\partial x}\left[\phi \Bigg / \frac{\partial \phi}{\partial x} \right](0,t) = \frac{\partial }{\partial x}\left[\phi \Bigg / \frac{\partial \phi}{\partial x} \right](L,t) = 0, \ \ t \ge 0.
\end{eqnarray}
\end{proposition}
\begin{proof}
Let us consider the linear wave Equation (\ref{wave}).
We performed the transformation
\begin{equation}\label{t1_6}
u = \phi \Bigg / \frac{\partial \phi}{\partial x}.
\end{equation}

The transformation (\ref{t1_6}) transforms Equation (\ref{wave}) to Equation (\ref{eq1_6}).
Equation~(\ref{lhe}) has the solutions  (\ref{wavesol}), (\ref{wavesol2}) and (\ref{wavesol3}). Thus, 
Equation (\ref{eq1_6}) has the solutions (\ref{sol1_eq16}),  (\ref{sol2_eq16}) and (\ref{sol3_eq16}).
\end{proof}

Note that (\ref{t1_6}) leads to $\phi = \exp \left( \int dx \ (1/u) \right)$.
\subsection{Transformations for the Heat Equation}
The solutions of the linear heat Equation (\ref{lhe}) are presented in the Appendix \ref{appa}.
\begin{proposition}
The equation
\begin{equation}\label{eq2}
\frac{\partial \phi}{\partial t} - a^2 \left(\frac{\partial \phi}{\partial x} \right)^2 - a^2 \frac{\partial^2 \phi}{\partial x^2} =0,
\end{equation}
has the solutions 
\begin{eqnarray}\label{sol1_eq2}
\phi(x,t) = \ln \Bigg( \frac{1}{2a \sqrt{\pi t}} \Bigg) + \ln \Bigg 
\{ \int \limits_{-\infty}^{+\infty}
d \xi \ f(\xi) \exp \left[  - \frac{(\xi - x)^2}{4 a t^2}\right] \Bigg \},
\end{eqnarray}
for the case $-\infty <x < \infty$ and for the initial condition
$\phi(x,0)=\ln[f(x)]$, and the solution
\begin{eqnarray}\label{sol2_eq2}
\phi(x,t) = \ln \Bigg \{ A + \frac{B-A}{L} + \sum \limits_{n=1}^{\infty} a_n \exp
\Bigg( - \frac{\pi^2 a^2 n^2}{L^2} t\Bigg ) \sin \Bigg(\frac{n \pi x}{L} \Bigg) \nonumber \\
a_n = - \frac{2}{n \pi}[A+(-1)^{n+1}B] + \frac{2}{L} \int \limits_0^L dx  \ f(x) \sin \Bigg( \frac{n \pi x}{L} \Bigg) \Bigg \}, 
\end{eqnarray}
for the case  $0<x<L$ and initial and boundary conditions $\phi(x,0)=\ln[f(x)]$
and $\phi(0,t)=\ln(A)$, $\phi(L,t)=ln(B)$.
\end{proposition}
\begin{proof}
Let us consider the linear heat Equation (\ref{lhe}).
We use the transformation
\begin{equation}\label{t2}
u = \exp(\phi).
\end{equation}

The transformation (\ref{t1}) transforms Equation (\ref{lhe}) to
Equation (\ref{eq2}).
Equation~(\ref{lhe}) has the solutions  (\ref{sol_heat1}) and (\ref{sol_heat2}). Thus,
Equation (\ref{eq2}) has the solutions (\ref{sol1_eq2}) and (\ref{sol2_eq2}).
\end{proof}
We can make the transformation more complicated. This more complicated
transformation will transform the linear heat equation into a nonlinear
Burgers equation. The transformation is the inverse transformation of
the Hopf--Cole transformation.
\begin{proposition}
The equation
\begin{equation}\label{eq2a}
\frac{\partial \phi}{\partial t} + \phi \frac{\partial \phi}{\partial x}  + a^2 \frac{\partial^2 \phi}{\partial x^2} =0,
\end{equation}
has the solutions 
\begin{eqnarray}\label{sol1_eq2xx}
\phi(x,t) = -2a^2 \frac{\frac{\partial}{\partial x}\int \limits_{-\infty}^{+\infty}
d \xi \ f(\xi) \exp \left[ - \frac{(\xi - x)^2}{4 a t^2}\right]}{\int \limits_{-\infty}^{+\infty}
d \xi \ f(\xi) \exp \left[ - \frac{(\xi - x)^2}{4 a t^2}\right]},
\end{eqnarray}
for the case $-\infty <x < \infty$ and for the initial condition
$\phi(x,0)=-2 \frac{a^2}{f}\frac{d f}{d x}$, and the solution
\begin{eqnarray}\label{sol2_eq2xx}
\phi(x,t) =  -2a^2 \frac{ \sum \limits_{n=1}^{\infty} \frac{n \pi a_n}{L} \exp
\Bigg( - \frac{\pi^2 a^2 n^2}{L^2} t \Bigg ) \sin \Bigg(\frac{n \pi x}{L} \Bigg)}{A + \frac{B-A}{L} + \sum \limits_{n=1}^{\infty} a_n \exp
\Bigg( - \frac{\pi^2 a^2 n^2}{L^2} t\Bigg ) \sin \Bigg(\frac{n \pi x}{L} \Bigg)},
\end{eqnarray}
for the case  $0<x<L$ and initial and boundary conditions $\phi(x,0)=-2 \frac{a^2}{f}\frac{df}{d x}$
and $\phi(0,t)=0$, $\phi(L,t)=0$.
\end{proposition}

\begin{proof}
Let us consider Equation (\ref{lhe}). We apply the transformation
\begin{equation}\label{t2_1}
u = \exp \left(- \frac{1}{2a^2} \int dx \ \phi (x,t)  \right).
\end{equation}

The transformation transforms (\ref{lhe}) to Equation (\ref{eq2a}).
Thus, Equation (\ref{eq2a}) has solutions (\ref{sol1_eq2xx}) and
(\ref{sol2_eq2xx}).
\end{proof}
We note that the transformation (\ref{t2_1}) leads to
\begin{equation}\label{hc}
\phi = - 2 a^2 \frac{1}{u} \frac{\partial u}{\partial x},
\end{equation}
which is the Hopf--Cole transformation.
\begin{proposition}
The equation
\begin{equation}\label{eq2b}
\phi \frac{\partial \phi}{\partial t} + \phi^3 \frac{\partial \phi}{\partial x} + a^2 \left( \frac{\partial \phi}{\partial x}\right)^2 + 
a^2 \phi \frac{\partial^2 \phi}{\partial x^2} =0,
\end{equation}
has the solutions 
\begin{eqnarray}\label{sol1_eq3}
\phi(x,t) = \phi(x,t) = \left \{ 2a^2 \frac{\frac{\partial}{\partial x}\int \limits_{-\infty}^{+\infty}
d \xi \ f(\xi) \exp \left[ - \frac{(\xi - x)^2}{4 a t^2}\right]}{\int \limits_{-\infty}^{+\infty}
d \xi \ f(\xi) \exp \left[ - \frac{(\xi - x)^2}{4 a t^2}\right]} \right \}^{1/2},
\end{eqnarray}
for the case $-\infty <x < \infty$ and for the initial condition
$\phi(x,0)=\left[ 2a^2 \frac{1}{f(x)} \frac{df}{dx}\right]^{1/2}$, and the solution
\begin{eqnarray}\label{sol2_eq3}
\phi(x,t) =  \phi(x,t) =  \left \{ 2a^2 \frac{ \sum \limits_{n=1}^{\infty} \frac{n \pi a_n}{L} \exp
\Bigg( - \frac{\pi^2 a^2 n^2}{L^2} t \Bigg ) \sin \Bigg(\frac{n \pi x}{L} \Bigg)}{A + \frac{B-A}{L} + \sum \limits_{n=1}^{\infty} a_n \exp
\Bigg( - \frac{\pi^2 a^2 n^2}{L^2} t\Bigg ) \sin \Bigg(\frac{n \pi x}{L} \Bigg)} \right \}^{1/2},
\end{eqnarray}
for the case  $0<x<L$ and initial and boundary conditions $\phi(x,0)=\left[ 2a^2 \frac{1}{f(x)} \frac{df}{dx}\right]^{1/2}$
and $\phi(0,t)=0$, $\phi(L,t)=0$.
\end{proposition}

\begin{proof}
Let us consider Equation (\ref{lhe}). Equation (\ref{eq2b})
has the solution
\begin{equation}\label{t2_3}
\phi = \left(2a^2 \frac{1}{u} \frac{\partial u}{\partial x} \right)^{1/2},
\end{equation}
and, thus, we arrive at the solutions (\ref{sol1_eq3}) and (\ref{sol2_eq3}).
\end{proof}
\begin{proposition}
The equation
\begin{equation}\label{eq2_3}
(\alpha-1) a^2 \left( \frac{\partial \phi}{\partial x} \right)^2
- a \phi \frac{\partial^2 \phi}{\partial x^2} -  \phi \frac{\partial \phi}{\partial t} =0,
\end{equation}
has the solutions 
\begin{eqnarray}\label{sol1_eq23}
\phi(x,t) = \Bigg\{ \frac{1}{2a \sqrt{\pi t}} \int \limits_{-\infty}^{+\infty}
d \xi \ f(\xi) \exp \left[  - \frac{(\xi - x)^2}{4 a t^2}\right] \Bigg \}^{1/\alpha},
\end{eqnarray}
for the case $-\infty <x < \infty$ and for the initial condition
$\phi(x,0)=[f(x)]^{1/\alpha}$, and the solution
\begin{eqnarray}\label{sol2_eq23}
\phi(x,t) =  \Bigg \{ A + \frac{B-A}{L} + \sum \limits_{n=1}^{\infty} a_n \exp
\Bigg( - \frac{\pi^2 a^2 n^2}{L^2} t\Bigg ) \sin \Bigg(\frac{n \pi x}{L} \Bigg) \nonumber \\
a_n = - \frac{2}{n \pi}[A+(-1)^{n+1}B] + \frac{2}{L} \int \limits_0^L dx  \ f(x) \sin \Bigg( \frac{n \pi x}{L} \Bigg) \Bigg \}^{1/\alpha}, 
\end{eqnarray}
for the case  $0<x<L$ and initial and boundary conditions $\phi(x,0)=[f(x)]^{1/\alpha}$
and \mbox{$\phi(0,t)=A^{1/\alpha}$}, $\phi(L,t)=B^{1/\alpha}$.
\end{proposition}
\begin{proof}
Let us consider the linear heat Equation (\ref{lhe}).
We use the transformation
\begin{equation}\label{t2_3x}
u = \phi^\alpha, \alpha \ne 1.
\end{equation}

The transformation (\ref{t2_3x}) transforms Equation (\ref{lhe}) to Equation (\ref{eq2_3}).
Equation~(\ref{lhe}) has the solutions  (\ref{sol_heat1}) and (\ref{sol_heat2}). Thus,
Equation (\ref{eq2_3}) has the solutions (\ref{sol1_eq23}) and (\ref{sol2_eq23}).
\end{proof}
\begin{proposition}
The equation
\begin{equation}\label{eq2_5}
 a^2 \frac{\partial^3 \phi}{\partial x^3} + 2 a \phi \frac{\partial \phi}{\partial x}
\frac{\partial^2 \phi}{\partial x^2} + 
\alpha^2 \frac{\partial \phi}{\partial x} \left(\frac{\partial \phi}{\partial x} \right)^2  + a^2 \frac{\partial \phi}{\partial x} -
\phi \left( \frac{\partial \phi}{\partial t} \right)^2 - 
\phi \frac{\partial^2 \phi}{\partial t} =0,
\end{equation}
has the solutions 
\begin{eqnarray}\label{sol1_eq25}
\phi(x,t) = 2^{1/2} \Bigg \{ \int dx \   \Bigg\{ \frac{1}{2a \sqrt{\pi t}} \int \limits_{-\infty}^{+\infty}
d \xi \ f(\xi) \exp \left[  - \frac{(\xi - x)^2}{4 a t^2}\right] \Bigg \} \Bigg \}^{1/2},
\end{eqnarray}
for the case $-\infty <x < \infty$ and for the initial condition
$\phi(x,0)=2^{1/2} [\int dx \  f(x)]^{1/2}$, and the solution
\begin{eqnarray}\label{sol2_eq25}
\phi(x,t) =  2^{1/2} \Bigg \{ \int dx \Bigg \{ A + \frac{B-A}{L} + \sum \limits_{n=1}^{\infty} a_n \exp
\Bigg( - \frac{\pi^2 a^2 n^2}{L^2} t\Bigg ) \sin \Bigg(\frac{n \pi x}{L} \Bigg) \Bigg \} \Bigg \}^{1/2} \nonumber \\
a_n = - \frac{2}{n \pi}[A+(-1)^{n+1}B] + \frac{2}{L} \int \limits_0^L dx  \ f(x) \sin \Bigg( \frac{n \pi x}{L} \Bigg), 
\end{eqnarray}
for the case  $0<x<L$ and initial and boundary conditions $\phi(x,0)= 2^{1/2}[\int dx \ f(x)]^{1/2}$
and $\phi(0,t)=2^{1/2} [\int dx A]^{1/2}$, $\phi(L,t)=2^{1/2} (\int dx  B)^{1/2}$.
\end{proposition}
\begin{proof}
Let us consider the linear heat Equation (\ref{lhe}).
We use the transformation
\begin{equation}\label{t2_5}
u =\phi \frac{\partial \phi}{\partial x}.
\end{equation}

The transformation (\ref{t2_5}) transforms Equation (\ref{lhe}) to
Equation (\ref{eq2_5}).
Equation (\ref{lhe}) has the solutions  (\ref{sol_heat1}) and (\ref{sol_heat2}). Thus,
Equation (\ref{eq2_5}) has the solutions (\ref{sol1_eq25}) and (\ref{sol2_eq25}).
\end{proof}
\begin{proposition}
The equation
\begin{equation}\label{eq2_6}
a^2 \left( \frac{\partial \phi}{\partial x} \right)^3  + 2 a^2 \phi
\frac{\partial^2 \phi}{\partial x^2} + a^2
\frac{\partial \phi}{\partial x} \frac{\partial^2 \phi}{\partial x^2} +
a^2 \frac{\partial^3 \phi}{\partial x^3}  -
\frac{\partial \phi}{\partial x} \frac{\partial \phi}{\partial t} -  \frac{\partial^2 \phi}{\partial x \partial t}=0,
\end{equation}
has the solutions 
\begin{eqnarray}\label{sol1_eq26}
\phi(x,t) = \ln \Bigg \{ \int dx \   \Bigg\{ \frac{1}{2a \sqrt{\pi t}} \int \limits_{-\infty}^{+\infty}
d \xi \ f(\xi) \exp \left[  - \frac{(\xi - x)^2}{4 a t^2}\right] \Bigg \} \Bigg \},
\end{eqnarray}
for the case $-\infty <x < \infty$ and for the initial condition
$\phi(x,0)= [\int dx \  f(x)]$, and the solution
\begin{eqnarray}\label{sol2_eq26}
\phi(x,t) =    \int dx \Bigg \{ A + \frac{B-A}{L} + \sum \limits_{n=1}^{\infty} a_n \exp
\Bigg( - \frac{\pi^2 a^2 n^2}{L^2} t\Bigg ) \sin \Bigg(\frac{n \pi x}{L} \Bigg) \Bigg \} \nonumber \\
a_n = - \frac{2}{n \pi}[A+(-1)^{n+1}B] + \frac{2}{L} \int \limits_0^L dx  \ f(x) \sin \Bigg( \frac{n \pi x}{L} \Bigg), 
\end{eqnarray}
for the case  $0<x<L$ and initial and boundary conditions $\phi(x,0)= \int dx \ f(x)]$
and\linebreak   $\phi(0,t)=2^{1/2} [\int dx A]^{1/2}$, $\phi(L,t)=2^{1/2} (\int dx  B)^{1/2}$.
\end{proposition}
\begin{proof}
Let us consider the linear heat Equation (\ref{lhe}).
We use the transformation
\begin{equation}\label{t2_6}
u =\exp(\phi) \frac{\partial \phi}{\partial x}.
\end{equation}

The transformation (\ref{t2_6}) transforms Equation (\ref{lhe}) to Equation (\ref{eq2_6}).
Equation (\ref{lhe}) has the solutions  (\ref{sol_heat1}) and (\ref{sol_heat2}). Thus, 
Equation (\ref{eq2_6}) has the solutions (\ref{sol1_eq26}) and (\ref{sol2_eq26}).
\end{proof}
\subsection{Transformation for the Laplace Equation}
The used solution of the Laplace Equation (\ref{lapl}) is given in Appendix \ref{appa}.
\begin{proposition}
The equation
\begin{equation}\label{eq3}
\left(\frac{\partial \phi}{\partial x} \right)^2 +
\left(\frac{\partial \phi}{\partial y} \right)^2 +
\frac{\partial^2 \phi}{\partial x^2} + \frac{\partial^2 \phi}{\partial y^2} =0,
\end{equation}
for the case of the rectangle domain $a<x<b$, $c<x<d$ and boundary conditions
$\phi(a,y)=\ln[f(y)]$; $\phi(b,y)=\ln[g(y)]$; $\phi(x,c)=\ln[h(c)]$; $\phi(x,d) = \ln[k (x)]$, has the solution
\begin{equation}\label{sol_eq3}
\phi = \ln(u_1 + u_2); \ \ \ u_1 = \sum \limits_n X_n(x) Y_n(y); \ \ \ u_2 = \sum \limits_m Z_m(x) V_m(y) .
\end{equation}
\end{proposition}

\begin{proof}
Let us consider the linear Laplace equation (\ref{lapl})
We make the transformation
\begin{equation}\label{t3}
u = \exp(\phi).
\end{equation}

The transformation (\ref{t3}) transforms Equation (\ref{lapl}) to
Equation (\ref{eq3}).
Equation (\ref{lapl}) has the solution  (\ref{sol_lapl}) Thus,
Equation (\ref{eq3}) has the solutions (\ref{sol_eq3}).
\end{proof}
\begin{proposition}
The equation
\begin{equation}\label{eq3_1}
\phi \frac{\partial^3 \phi}{\partial x^3} + 
3 \frac{\partial \phi}{\partial x} \frac{\partial^2 \phi}{\partial x^2} + \frac{\partial \phi}{\partial x} \frac{\partial^2 \phi}{\partial y^2} + 
2 \frac{\partial \phi}{\partial y} \frac{\partial^2 \phi}{\partial x \partial y} + \phi \frac{\partial^3 \phi}{\partial x \partial y^2} =0,
\end{equation} 
for the
case of the rectangle domain $a<x<b$, $c<x<d$ and boundary conditions\linebreak  
$\phi(a,y)=2^{1/2}\{\int dx [f(y)]\}^{1/2}$; $\phi(b,y)=2^{1/2}\{ \int dx [g(y)]\}^{1/2}$; $\phi(x,c)=2^{1/2} \{ \int dx [h(c)]\}^{1/2}$; $\phi(x,d) =2^{1/2} \{ \int dx [k (x)]\}^{1/2}$, has the solution
\begin{equation}\label{sol_eq31}
\phi = 2^{1/2} \left \{ \int dx (u_1 + u_2) \right \}^{1/2}; \ \ \ u_1 = \sum \limits_n X_n(x) Y_n(y); \ \ \ u_2 = \sum \limits_m Z_m(x) V_m(y).
\end{equation}
\end{proposition}

\begin{proof}
Let us consider the linear Laplace Equation (\ref{lapl}).
We use the transformation
\begin{equation}\label{t3_1}
u = \phi \frac{\partial \phi}{\partial x}.
\end{equation}

The transformation (\ref{t3_1}) transforms Equation (\ref{lapl}) to Equation (\ref{eq3_1}).
Equation (\ref{lapl}) has the solution  (\ref{sol_lapl}) Thus,
Equation (\ref{eq3_1}) has the solution (\ref{sol_eq31}).
\end{proof}
\begin{proposition}
The equation
\begin{equation}\label{eq3_2}
(\alpha -1 ) \left[ \left(\frac{\partial \phi}{\partial x} \right)^2 +
\left(\frac{\partial \phi}{\partial y} \right)^2 \right] + \phi \left(
\frac{\partial^2 \phi}{\partial x^2} + \frac{\partial^2 \phi}{\partial y^2} \right) =0,
\end{equation} for the
case of the rectangle domain $a<x<b$, $c<x<d$ and boundary conditions
$\phi(a,y)=[f(y)]^{1/\alpha}$; $\phi(b,y)=[g(y)]^{1/\alpha}$; $\phi(x,c)=[h(c)]^{1/\alpha}$; $\phi(x,d) = [k (x)]^{1/\alpha}$, has the solution
The solution is
\begin{equation}\label{sol_eq32}
\phi = (u_1 + u_2)^{1/\alpha}; \ \ \ u_1 = \sum \limits_n X_n(x) Y_n(y); \ \ \ u_2 = \sum \limits_m Z_m(x) V_m(y) .
\end{equation}
\end{proposition}

\begin{proof}
Let us consider the linear Laplace Equation (\ref{lapl}).
We make the transformation
\begin{equation}\label{t3_2}
u = \phi^\alpha, \alpha \ne 1,
\end{equation}
where $\alpha$ =const. 
The transformation (\ref{t3_2}) transforms Equation (\ref{lapl}) to Equation (\ref{eq3_2}).
Equation (\ref{lapl}) has the solution  (\ref{sol_lapl}) Thus,
Equation (\ref{eq3_2}) has the solution (\ref{sol_eq32}).
\end{proof}
\section{Transformations of Linear and Nonlinear Ordinary Differential Equations\label{sec3}}
The main idea above was to start from linear partial differential equations
and by means of appropriate transformations to obtain solutions of
nonlinear differential equations.

We can start from linear ordinary differential equations and
obtain solutions of nonlinear ordinary differential equations.
Let us consider, for  example, the linear ordinary differential equation
\begin{equation}\label{ode1}
\frac{du}{dx} + p(x) u = q(x).
\end{equation}

The solution of this equation is
\begin{equation}\label{sol_ode1}
u(x) = \exp \left( - \int dx \ p(x) \right)\left[ C + \int dx 
\ q(x) \exp \left( \int dx \ p(x) \right)\right].
\end{equation}

$C$ in (\ref{sol_ode1}) is a constant of integration. 
The Cauchy problem $u(x_0)=u_0$ for (\ref{ode1}) has the~solution
\begin{equation}\label{cauchy_ode1}
u(x) = \exp \left( - \int \limits_{x_0}^x dy \ p(y) \right)\left[ u_0 + \int \limits_{x_0}^x dy 
\ q(x) \exp \left( \int \limits_{x_0}^x dy \ p(y) \right)\right].
\end{equation}

\begin{proposition}
The nonlinear differential equation
\begin{equation}\label{nde1}
\phi \frac{d^2 \phi}{dx^2} + \left ( \frac{\partial \phi}{\partial x} \right)^2 + p(x) \phi \frac{d \phi}{dx} = q(x),
\end{equation}
has the solution
\begin{equation}\label{sol_nde1}
\phi = 2^{1/2} \left \{ \int dx \left \{ \exp \left( - \int dx \ p(x) \right)\left[ C + \int dx 
\ q(x) \exp \left( \int dx \ p(x) \right)\right] \right \} \right \}^{1/2}.
\end{equation}

The Cauchy problem $u(x_0)=u_0$ for (\ref{nde1}) has the solution
\begin{equation}\label{cauchy_nde1}
u(x) = 2^{1/2} \left \{\exp \left( - \int \limits_{x_0}^x dy \ p(y) \right)\left[ u_0 + \int \limits_{x_0}^x dy 
\ q(x) \exp \left( \int \limits_{x_0}^x dy \ p(y) \right)\right] \right \}^{1/2}.
\end{equation}
\end{proposition}
\begin{proof}
Let us apply the transformation
\begin{equation}\label{tx1}
u = \phi(x) \frac{d \phi}{dx}.
\end{equation}
(\ref{tx1}) transforms (\ref{ode1}) to (\ref{nde1}) and the solutions (\ref{sol_ode1}) and (\ref{cauchy_ode1}) are transformed to (\ref{sol_nde1}) and (\ref{cauchy_nde1}).
\end{proof}
\par
\begin{proposition}
Let us consider the equation
\begin{equation}\label{nde2}
\frac{d^2 \phi}{d x^2} + \left( \frac{d \phi}{d x} \right)^2 +
p(x) \frac{d \phi}{d x} - q(x) \exp(- \phi) = 0.  
\end{equation}

The solution of (\ref{nde2}) is
\begin{equation}\label{sol_nde2}
\phi(x) = \ln \left \{ \int dx \left \{ \exp \left( - \int dx \ p(x) \right)\left[ C + \int dx 
\ q(x) \exp \left( \int dx \ p(x) \right)\right]\right \} \right \}.
\end{equation}

The Cauchy problem $u(x_0)=u_0$ for (\ref{nde2}) has the solution
\begin{equation}\label{cauchy_nde2}
u(x) = \ln \left \{ \int dx \left \{ \left( - \int \limits_{x_0}^x dy \ p(y) \right)\left[ u_0 + \int \limits_{x_0}^x dy 
\ q(x) \exp \left( \int \limits_{x_0}^x dy \ p(y) \right)\right] \right \} \right \}.
\end{equation}
\end{proposition}
\begin{proof}
Let us now apply the transformation
\begin{equation}\label{tx2}
u = \exp(\phi) \frac{d \phi}{dx}.
\end{equation}

The transformation (\ref{tx2}) transforms (\ref{ode1}) to (\ref{nde2}). The solutions
(\ref{sol_ode1}) and (\ref{cauchy_ode1}) are transformed to (\ref{sol_nde2}) and (\ref{cauchy_nde2}).
\end{proof}

Our last example here will be connected to the equation
\begin{equation}\label{ode3}
\frac{d^2 u}{dx^2} + a \frac{du}{dx} + b u = 0.
\end{equation}

The solution of (\ref{ode3}) depends on the characteristic numbers $\lambda_{1,2} = - \frac{a}{2} \mp \sqrt{\frac{a^2}{4} - b}$. The solution is
\begin{equation}\label{ode3_sol1}
u = C_1 \exp(\lambda_1 x) + C_2 \exp(\lambda_2 x), \ \ \lambda_1 \ne \lambda_2,
\end{equation}
and 
\begin{equation}\label{ode3_sol2}
u = (C_1  + C_2 x) \exp (\lambda x), \ \ \lambda_1 = \lambda_2 = \lambda.
\end{equation}

$C_{1,2}$ are constants.
Let us apply the transformation $u = \phi^{\alpha}$. (\ref{ode3}) is
transformed to
\begin{equation}\label{nde3_1}
\alpha (\alpha - 1) \left( \frac{d \phi}{d x} \right)^2 + \alpha \phi
\frac{d^2 \phi}{d x^2} + \alpha a \phi \frac{d \phi}{dx} + b \phi =0.
\end{equation}

The solutions of (\ref{nde3_1}) are
\begin{equation}\label{nde31_sol1}
 \phi = [C_1 \exp(\lambda_1 x) + C_2 \exp(\lambda_2 x)]^{1/\alpha}, \ \ \lambda_1 \ne \lambda_2,
\end{equation}
and
\begin{equation}\label{nde31_sol2}
\phi = [(C_1  + C_2 x) \exp (\lambda x)]^{1/\alpha}, \ \ \lambda_1 = \lambda_2 = \lambda.
\end{equation}

\begin{proposition}
Let us consider the equation
\begin{equation}\label{nde3_2}
\frac{d^3 \phi}{d x^3} + 3 \frac{d \phi}{d x} \frac{d^2 \phi}{d x^2} + \left( \frac{d \phi}{d x} \right)^3 + a \frac{d^2 \phi}{d x^2} + a
\left( \frac{d \phi}{d x} \right)^2 + b \exp(\phi) \frac{d \phi}{d x} =0.
\end{equation}
(\ref{nde3_2}) has the solutions
\begin{equation}\label{nde32_sol1}
\phi =  \ln \{  \int dx [C_1 \exp(\lambda_1 x) + C_2 \exp(\lambda_2 x)] \}, \ \ \lambda_1 \ne \lambda_2,
\end{equation}
and
\begin{equation}\label{nde32_sol2}
\phi = \ln \{\int dx [  (C_1  + C_2 x) \exp (\lambda x)] \}, \ \ \lambda_1 = \lambda_2.
\end{equation}
\end{proposition}
\begin{proof}
We  apply the transformation $u = \exp(\phi) \frac{d \phi}{dx}$ to (\ref{ode3}).
(\ref{ode3}) is transformed to Equation (\ref{nde3_1}) and solutions (\ref{nde31_sol1}) and (\ref{nde31_sol2}) are transformed to (\ref{nde32_sol1}) and (\ref{nde32_sol2}).
\end{proof}

\begin{proposition}
The equation
\begin{equation}\label{nde3_3}
\frac{d^2 \phi}{d x^2} + \left(\frac{d \phi}{dx} \right)^2 + a \frac{d \phi}{dx} + b \exp(\phi) = 0,
\end{equation}
has the solutions
\begin{equation}\label{nde33_sol1}
\phi = \ln [ C_1 \exp(\lambda_1 x) + C_2 \exp(\lambda_2 x)], \ \ \lambda_1 \ne \lambda_2,
\end{equation}
and
\begin{equation}\label{nde33_sol2}
\phi = \ln [(C_1  + C_2 x) \exp (\lambda x)], \ \ \lambda_1 = \lambda_2 = \lambda.
\end{equation}
\end{proposition}
\begin{proof}
We make the transformation $u=\exp(\phi)$. Equation (\ref{ode3}) is transformed to
(\ref{nde3_3}). Solutions (\ref{nde31_sol1}) and (\ref{nde31_sol2}) are transformed to solutions (\ref{nde33_sol1}) and (\ref{nde33_sol2}).
\end{proof}

\par 
We can start from nonlinear ordinary differential equations and we can 
obtain solutions of other nonlinear differential equations. Let us
consider as examples the equations of Bernoulli and Riccati. The 
equation of Bernoulli
\begin{equation}\label{bern}
\frac{d \phi}{dx} + p(x) \phi = q(x) \phi^m, \ m \ne 0, \ m \ne 1,
\end{equation}
is an example of a nonlinear equation which can be obtained from a
linear equation by means of a transformation. The linear equation is
\begin{equation}\label{lin_bern}
\frac{du}{dx} + (1-m) p(x) u = (1-m) q(x),
\end{equation}
and the transformation is $u=\phi^{1-m}$.  The solution of the linear equation is
\begin{equation}\label{sol_linbern}
u(x) = \exp \left [-(1-m) \int dx p(x) \right] \left \{ C + (1-m)
 \int dx \ q(x) \exp \left[(1-m) \int dx \ p(x) \right]  \right\}.
\end{equation}     

Thus, the solution of the Bernoulli equation is

\begin{equation}\label{sol_bern}
\phi = \left \{  \exp \left [-(1-m) \int dx \ p(x) \right] \left \{ C + (1-m)
 \int dx \ q(x) \exp \left[(1-m) \int dx \ p(x) \right]  \right\} \right \}^{1/(1-m)}.
\end{equation}

\begin{proposition}
The equation
\begin{equation}\label{bern_1}
\frac{d \psi}{d x} + p(x) - q(x) \exp[(m-1)\psi] = 0,
\end{equation}
has the solution
\begin{eqnarray}\label{sol_bern1}
\psi = \ln \Bigg \{ \Bigg \{  \exp \left [-(1-m) \int dx \ p(x) \right] \Bigg \{ C + (1-m) \times \nonumber \\
 \int dx \ q(x) \exp \left[(1-m) \int dx \ p(x) \right]  \Bigg\} \Bigg \}^{1/(1-m)}\Bigg \}.
\end{eqnarray}
\end{proposition}
\begin{proof}
Let us apply the transformation $\phi = \exp(\psi)$.
This transformation transforms the equation of Bernoulli (\ref{bern}) to Equation (\ref{bern_1}) and solution (\ref{sol_bern}) is transformed to~(\ref{sol_bern1}). 
\end{proof}

\begin{proposition}
The equation
\begin{equation}\label{bern_2}
\left( \frac{d \psi}{d x}\right)^2 + \frac{d^2 \psi}{d x^2} + p(x)\frac{d \psi}{d x} - q(x) \exp[(m-1)\psi]\left( \frac{d \psi}{d x} \right)^m = 0,
\end{equation}
has the solution
\begin{eqnarray}\label{sol_bern2}
\psi = \ln \Bigg \{ \int dx \Bigg \{ \Bigg \{  \exp \left [-(1-m) \int dx \ p(x) \right] \Bigg \{ C + (1-m) \times \nonumber \\
 \int dx \ q(x) \exp \left[(1-m) \int dx \ p(x) \right]  \Bigg\} \Bigg \}^{1/(1-m)}\Bigg \} \Bigg\}.
\end{eqnarray}
\end{proposition}
\begin{proof}
Let us apply the transformation $\phi = \exp(\psi) \frac{d \psi}{dx}$.
This transformation transforms the equation of Bernoulli (\ref{bern}) to Equation (\ref{bern_2}) and solution (\ref{sol_bern}) is transformed to~(\ref{sol_bern2}). 
\end{proof}
\par 
Next, we discuss the equation of Riccati. We shall consider the
specific case of the Riccati equation, which has the constant coefficients
\begin{equation}\label{ricc}
\frac{du}{dx} = \alpha_2 u^2 + \alpha_1 u + \alpha_0.
\end{equation}

The general solution of this equation is
\begin{equation}\label{sol_ricc}
u(t) =- \frac{\alpha_1}{2 \alpha_2} - \frac{\theta}{2 \alpha_2}
\tanh \left[ \frac{\theta(x+C)}{2} \right] + 
\frac{D}{\cosh^2\left[ \frac{\theta(x+C)}{2} \right] \left \{ E - 
\frac{2 \alpha_2 D}{\theta}
\tanh \left[ \frac{\theta(x+C)}{2} \right] \right \}},
\end{equation}
where $\theta^2 = \alpha_1^2 - 4 \alpha_0 \alpha_2$, and $C,D,E$ are
constants.
\begin{proposition}
The equation
\begin{equation}\label{ricc_1}
\frac{d \phi}{d x} = a_2 \exp(\phi) + a_1 + a_0 \exp(-\phi),
\end{equation}
has the solution

\begin{equation}\label{sol_ricc1}
\phi = \ln \Bigg \{ - \frac{\alpha_1}{2 \alpha_2} - \frac{\theta}{2 \alpha_2}
\tanh \left[ \frac{\theta(x+C)}{2} \right] + 
\frac{D}{\cosh^2\left[ \frac{\theta(x+C)}{2} \right] \left \{ E - 
\frac{2 \alpha_2 D}{\theta}
\tanh \left[ \frac{\theta(x+C)}{2} \right] \right \}}\Bigg \}.
\end{equation}
\end{proposition}
Note the following specific cases of (\ref{ricc_1}). Let $a_2 = a_0 = a/2$. Then, we obtain the equation
\begin{equation}\label{specx1}
\frac{\partial \psi}{\partial x} = a_1 +a \cosh (\psi).
\end{equation}

Let $a_2 - a_0 = a/2$. Then, we obtain the equation
\begin{equation}\label{specx2}
\frac{\partial \psi}{\partial x} = a_1 +a \sinh (\psi).
\end{equation}
\begin{proof}
Let us apply the transformation $\phi = \exp(\psi)$.
This transformation transforms the equation of Riccati (\ref{ricc}) to
Equation (\ref{ricc_1}) and solution (\ref{sol_ricc}) is transformed to~(\ref{sol_ricc1}). 
\end{proof}
\begin{proposition}
The equation
\begin{equation}\label{ricc_2}
\frac{d^2 \psi}{d x^2} + \left( \frac{d \psi}{d x} \right)^2 = a_2 \exp(\psi)\left( \frac{d \psi}{d x} \right)^2 + a_1 \frac{d \psi}{d x} + a_0 \exp(-\psi),
\end{equation}
has the solution

\begin{equation}\label{sol_ricc2}
\phi = \ln \Bigg \{ \int dx \Bigg \{ - \frac{\alpha_1}{2 \alpha_2} - \frac{\theta}{2 \alpha_2}
\tanh \left[ \frac{\theta(x+C)}{2} \right] + 
\frac{D}{\cosh^2\left[ \frac{\theta(x+C)}{2} \right] \left \{ E - 
\frac{2 \alpha_2 D}{\theta}
\tanh \left[ \frac{\theta(x+C)}{2} \right] \right \}}\Bigg \} \Bigg \}.
\end{equation}
\end{proposition}
\begin{proof}
Let us apply the transformation $\phi = \exp(\psi)\frac{d \psi}{dx}$.
This transformation transforms the equation of Riccati (\ref{ricc}) to Equation (\ref{ricc_2}) and solution (\ref{sol_ricc}) is transformed to~(\ref{sol_ricc2}).
\end{proof}
\begin{proposition}
The equation
\begin{equation}\label{ricc_3}
\alpha \frac{d \phi}{d x}  = a_2 \phi^{\alpha + 1} + a_1  \psi + a_0 \phi^{1 - \alpha},
\end{equation}
has the solution
\begin{equation}\label{sol_ricc3}
\phi = \Bigg \{ - \frac{\alpha_1}{2 \alpha_2} - \frac{\theta}{2 \alpha_2}
\tanh \left[ \frac{\theta(x+C)}{2} \right] + 
\frac{D}{\cosh^2\left[ \frac{\theta(x+C)}{2} \right] \left \{ E - 
\frac{2 \alpha_2 D}{\theta}
\tanh \left[ \frac{\theta(x+C)}{2} \right] \right \}}\Bigg \}^{1/\alpha}.
\end{equation}
\end{proposition}
\begin{proof}
Let us apply the transformation $\phi = \psi^\alpha$ ($\alpha \ne 1$).
This transformation transforms the equation of Riccati (\ref{ricc}) to Equation (\ref{ricc_3}) and solution (\ref{sol_ricc}) is transformed to~(\ref{sol_ricc3}). 
\end{proof}
\section{Transformations of Nonlinear Partial Differential Equations\label{sec4}}
We can also start  from nonlinear partial differential equations and we can obtain
solutions of other nonlinear partial differential equations by applying appropriate
transformations. Here, we consider as an example the Korteweg--de Vries equation
\begin{equation}\label{kdv}
\frac{\partial u}{\partial t} + \frac{\partial^3 u}{\partial x^3} - 6 u \frac{\partial u}{\partial x} =0.
\end{equation}

The single soliton solution of (\ref{kdv}) is given by
\begin{equation}\label{kdv_sol1}
u(x,t) = - \frac{c}{2} {\rm sech}^2 \left[\frac{\sqrt{c}}{2}(x-ct)  \right],
\end{equation}
where $c$ is the velocity of the wave. The
N-soliton solution of (\ref{kdv}) is given by
\begin{equation}\label{kdv_sol2}
u(x,t) = - 2 \frac{\partial^2}{\partial x^2} \ln \det[A(x,t)],
\end{equation}
where the matrix $A_{nm}$ is
$$
A_{nm}(x,t) = \delta_{nm} + \frac{\beta_n \exp(8 \xi_n^3 t) \exp[-(\xi_n + \xi_m) x]}{\xi_n + \xi_n}.
$$

Above, the parameters $\xi_1 \ge xi_2 \ge \dots, \ge \xi_N >0$ and
the parameters  $\beta_i$ , $i=1,\dots,N$ are nonzero ones.
\par 
The bisoliton solution of the Korteweg--de Vries equation which satisfies the initial condition $u(x,0)=$-$6  {\rm sech}^2{(x)}$ is
\begin{equation}\label{kdv_sol3}
u(x,t) = -12 \frac{3 + 4 \cosh(8t-2x)+\cosh(64t - 4x)}{[\cosh(36t - 3x)+3\cosh(28t-x)]^2}.
\end{equation}
\par 
\begin{proposition}
The equation
\begin{equation}\label{kdv_1}
\alpha^2 \phi^2 \frac{\partial ^3 \phi}{\partial x^3} + 3 \alpha(\alpha-1) \phi \frac{\partial \phi}{\partial x} \frac{\partial^2 \phi}{\partial x^2} + \alpha(\alpha-1)(\alpha-2)\left(\frac{\partial \phi}{\partial x} \right)^3 + \alpha \phi^2 \frac{\partial \phi}{\partial t} - 6 \phi^{\alpha+2}\frac{\partial \phi}{\partial x} =0,
\end{equation}
has the solutions 
\begin{equation}\label{sol1_kdv1}
\phi(x,t) = \left \{ - \frac{c}{2} {\rm sech}^2 \left[\frac{\sqrt{c}}{2}(x-ct)  \right] \right \}^{1/\alpha},
\end{equation}
and
\begin{equation}\label{sol2_kdv1}
\phi(x,t) = \left \{- 2 \frac{\partial^2}{\partial x^2} \ln \det[A(x,t)] \right \}^{1/\alpha},
\end{equation}
and
\begin{equation}\label{sol3_kdv1}
u(x,t) = \left \{-12 \frac{3+ 4 \cosh(8t-2x)+\cosh(64t-4x)}{[\cosh(36t-3x)+3\cosh(28t-x)]^2} \right \}^{1/\alpha}.
\end{equation}
\end{proposition}
\begin{proof}
We apply the transformation $u = \phi^{\alpha}$ ($\alpha$ = const and $\alpha \ne 1$) to Equation (\ref{kdv}). The transformation transforms (\ref{kdv})
to (\ref{kdv_1}). The solutions (\ref{kdv_sol1}), (\ref{kdv_sol2})
and (\ref{kdv_sol3})
are transformed to (\ref{sol1_kdv1}), (\ref{sol2_kdv1}) and (\ref{sol3_kdv1}). 
\end{proof} 

Figure \ref{fig7} shows the bisoliton solution of Equation (\ref{kdv_1}).\vspace{-6pt}
\begin{figure}[h]
\centering
\includegraphics[scale=0.45,angle=-90]{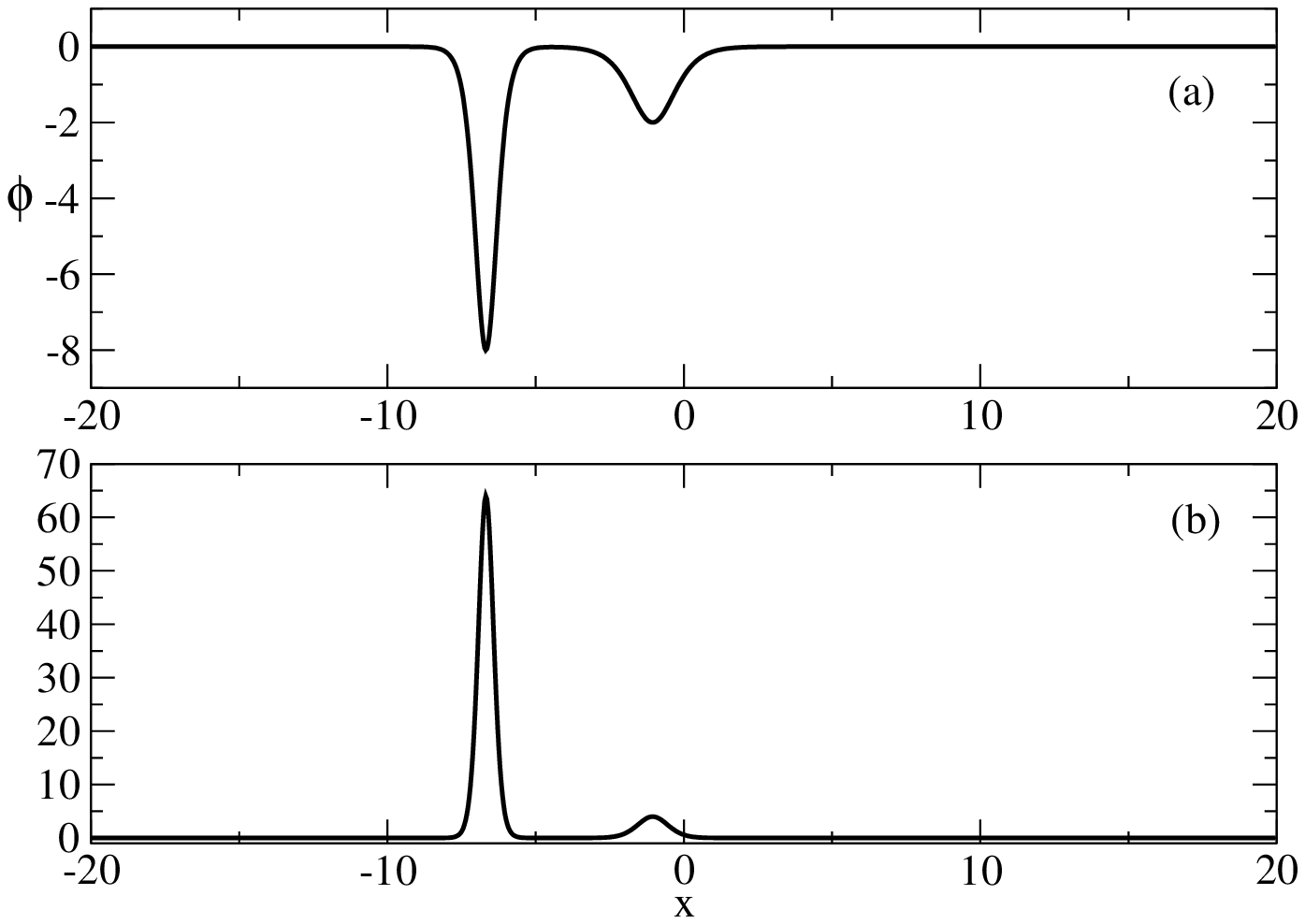}
\caption{The 
 bisoliton solution to Equation (\ref{kdv_1}). (\textbf{a}) The bisoliton solution of the Korteweg--de Vries equation. (\textbf{b}) The corresponding bisoliton solution of (\ref{kdv_1}).
The parameters of the solution are $c=2.0$, $\alpha = 1/2$.\label{fig7}}
\end{figure}

The soliton properties of the solution (\ref{sol3_kdv1}) are illustrated in Figure \ref{fig8}. We observe the motion of
two solitons having different velocities. The larger (and the faster) soliton moves through the smaller soliton
and continues to travel without change in its form. The smaller soliton also does not change its form.

\begin{figure}[htb!]
\centering 
\includegraphics[scale=0.9]{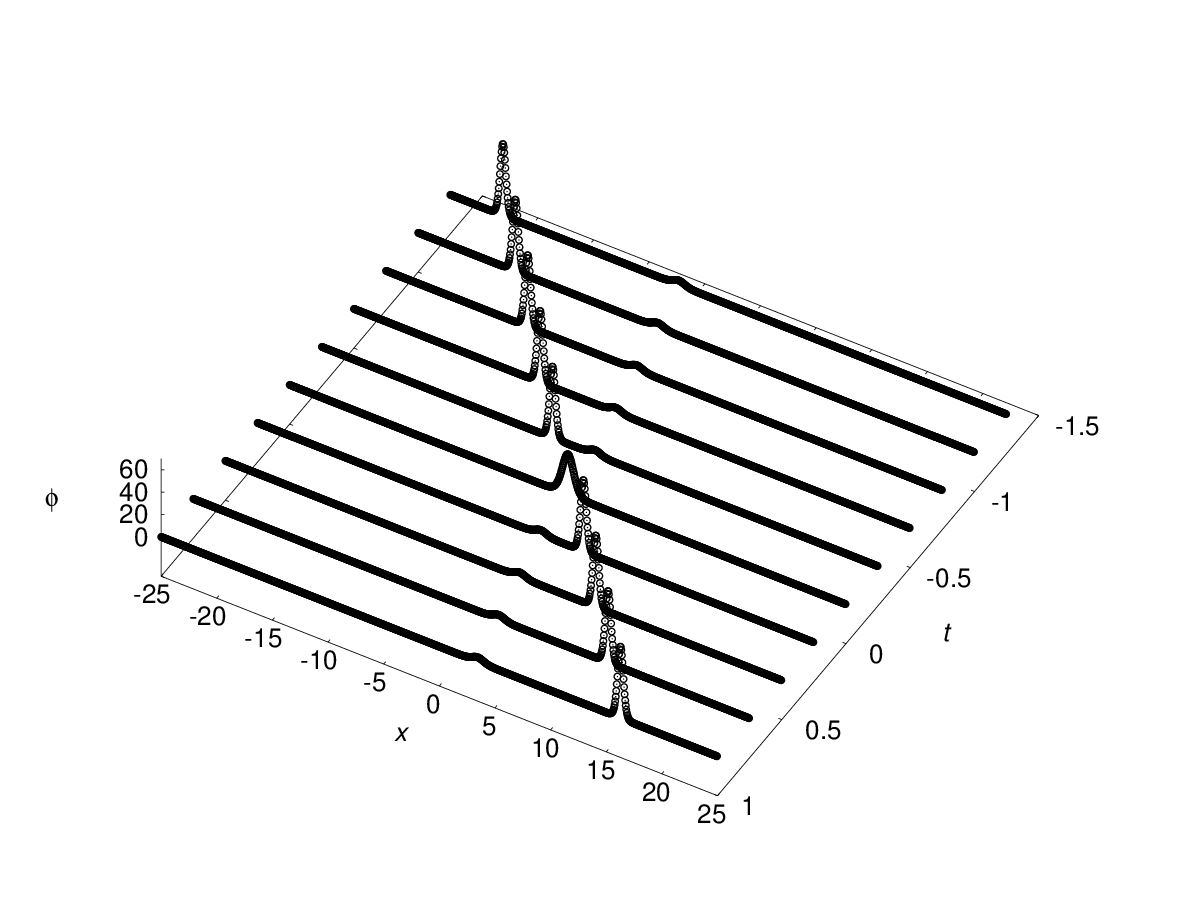}
\caption{Illustration 
 of the soliton properties of the waves connected to the  bisoliton solution  of (\ref{kdv_1}).
The parameters of the solution are $c=3.2$, $\alpha = 1/2$.\label{fig8}}
\end{figure}

\begin{proposition}
The equation
\begin{equation}\label{kdv_2}
\frac{\partial \phi}{\partial t}  + \left(\frac{\partial \phi}{\partial x} \right)^3 + 3 \frac{\partial \phi}{\partial x} + \frac{\partial^3 \phi}{\partial x^3} - 6 \exp(\phi)\frac{\partial \phi}{\partial x}=0,
\end{equation}
has the solutions 
\begin{equation}\label{sol1_kdv2}
\phi(x,t) = \ln \left \{ - \frac{c}{2} {\rm sech}^2 \left[\frac{\sqrt{c}}{2}(x-ct)  \right] \right \},
\end{equation}
and
\begin{equation}\label{sol2_kdv2}
\phi(x,t) = \ln \left \{- 2 \frac{\partial^2}{\partial x^2} \ln \det[A(x,t)] \right \},
\end{equation}
and
\begin{equation}\label{sol3_kdv2}
u(x,t) = \ln \left \{-12 \frac{3+ 4 \cosh(8t-2x)+\cosh(64t-4x)}{[\cosh(36t-3x)+3\cosh(28t-x)]^2} \right \}.
\end{equation}
\end{proposition}
\begin{proof}
We apply the transformation $u = \exp(\phi)$  to Equation (\ref{kdv}). The transformation transforms (\ref{kdv})
to (\ref{kdv_2}). The solutions (\ref{kdv_sol1}), (\ref{kdv_sol2})
and (\ref{kdv_sol3})
are transformed to (\ref{sol1_kdv2}), (\ref{sol2_kdv2}) and (\ref{sol3_kdv2}).
\end{proof} 
\par 
\begin{proposition}
The equation
\begin{equation}\label{kdv_3}
\phi^2 \frac{\partial \phi}{\partial t} + \phi^2 \frac{\partial^3 \phi}{\partial x^3}- 3 \phi \frac{\partial \phi}{\partial x} \frac{\partial^3 \phi}{\partial x^3} + 2 \left( \frac{\partial \phi}{\partial x} \right)^3 - 6 \ln(\phi) \phi \left( \frac{\partial \phi}{\partial x} \right)^2=0,
\end{equation}
has the solutions 
\begin{equation}\label{sol1_kdv3}
\phi(x,t) = \exp \left \{ - \frac{c}{2} {\rm sech}^2 \left[\frac{\sqrt{c}}{2}(x-ct)  \right] \right \},
\end{equation}
and
\begin{equation}\label{sol2_kdv3}
\phi(x,t) = \exp \left \{- 2 \frac{\partial^2}{\partial x^2} \ln \det[A(x,t)] \right \},
\end{equation}
and
\begin{equation}\label{sol3_kdv3}
\phi(x,t) = \exp \left \{-12 \frac{3+ 4 \cosh(8t-2x)+\cosh(64t-4x)}{[\cosh(36t-3x)+3\cosh(28t-x)]^2} \right \}.
\end{equation}
\end{proposition}
\begin{proof}
We apply the transformation $u = \ln(\phi)$  to Equation (\ref{kdv}). The transformation transforms (\ref{kdv})
to (\ref{kdv_3}). The solutions (\ref{kdv_sol1}), (\ref{kdv_sol2}) 
and (\ref{kdv_sol3})
are transformed to (\ref{sol1_kdv3}), (\ref{sol2_kdv3}) and (\ref{sol3_kdv3}).
\end{proof} 

Figure \ref{fig9} illustrates the soliton and the bisoliton solution to Equation
(\ref{kdv_3}).
\begin{figure}[h]
\centering
\includegraphics[scale=0.45,angle=-90]{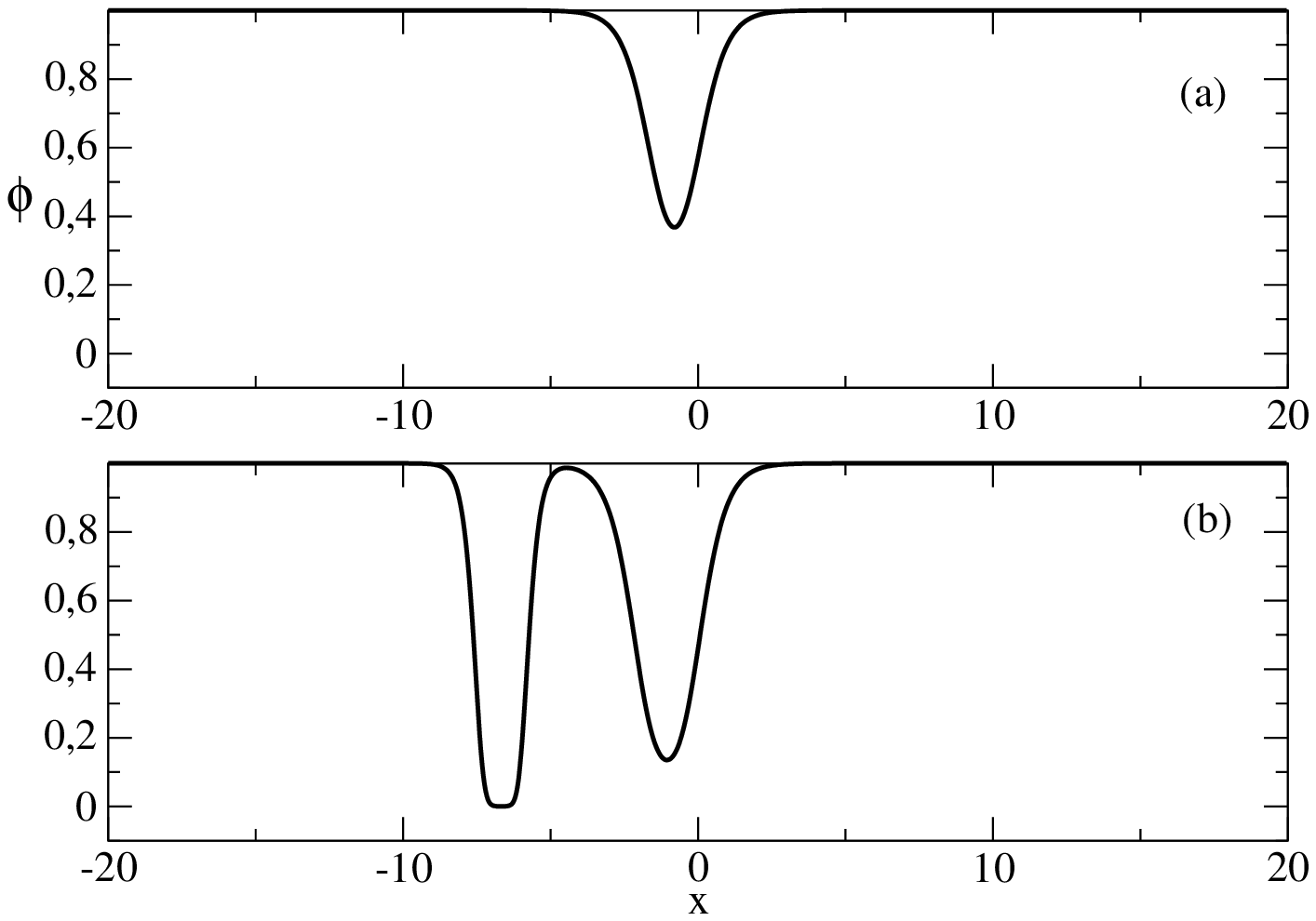}
\caption{The 
 soliton and the bisoliton solution to Equation (\ref{kdv_3}).
(\textbf{a}) The soliton solution. (\textbf{b})~The corresponding bisoliton solution of (\ref{kdv_3}).\label{fig9}}
\end{figure}

\section{Concluding Remarks\label{sec5}}

Transformations are very useful for obtaining exact solutions of nonlinear
differential equations. As examples, we mention the B{\"a}cklund transformation \cite{bk1}-\cite{bk5} and
the transformation of Darboux \cite{db1}-\cite{db10}. The transformation of B{\"a}cklund allows us to
obtain new exact solutions of appropriate equation if we know an exact solution of
this equation. The Darboux transformation is a simultaneous mapping between solutions
and coefficients to a pair of equations (or systems of equations) of the same form.
The methodology proposed in this article allows us to obtain exact solutions to
nonlinear differential equations if we know exact solutions of different (linear
or nonlinear) differential equations.

In this article, we study the possibility of obtaining exact solutions to nonlinear
differential equations by means of transformations applied to more simple
(linear or nonlinear) differential equations. The idea of this method of 
transformations is very simple: a transformation transforms a linear or nonlinear
differential equation with a known solution to a more complicated nonlinear differential
equation. The same transformation transforms the known solution of the more simple
equation to an exact solution to the  corresponding nonlinear equation.
A similar idea, for example, is used in the
simple equations method (SEsM). There, we have a simple equation which usually is
a nonlinear differential equation with a known solution. By means of this solution,
we construct a solution of a more complicated nonlinear differential equation. This
construction can be thought of as a transformation which transforms the solution of the
simple equation to the solution of the more complicated equation.
 
In this article, we demonstrate the result of the application of very simple
transformations on several linear and nonlinear differential equations. We
obtain solutions of nonlinear differential equations which are connected
to interesting phenomena: (a)~occurrence of a solitary wave--solitary antiwave from
a state which is equal to $0$ at the moment $t=0$; (b) splitting of solitary wave of two
solitary waves; (c) solitons which have different amplitude and the same velocity.
We stress result (c). Usually, solitons of different heights have different velocities.
This can be observed in Figure \ref{fig8}. There, the larger (and the faster) soliton passes through
the smaller (and the slower) soliton and continues to travel without change in its form.
Figure \ref{fig6} shows that we can have two solitons of different forms which travel with the same
velocity. We observe the collision of these solitons and after the collision they travel
further without change in their form.  Such a possible class of  solitons is very interesting.

One limitation of the discussed methodology is that it needs an equation with a known exact
solution in order to obtain exact solutions of other equations. This limitation is the same
as in the case of other transformation methodologies such as, for example, the methodology based on
the B{\"a}cklund transformation. Another limitation of the methodology is that we need to know
not only the transformation $u=u(\phi,\dots)$, but also the explicit form of the inverse transformation
$\phi=\phi(u,\dots)$ in order to construct the exact solution of the equation for $\phi$.

The methodology reported in this article can be applied to various problems. First of all,
one can study additional equations. Just one example is the class of equations connected
to the nonlinear Schr{\"o}dinger equation. Let us consider the equation \cite{zs}
\begin{equation}\label{nls1}
i \frac{\partial u}{\partial t} + \frac{\partial^2 u}{\partial x^2} + \kappa \mid u \mid^2 u = 0,
\end{equation}
where $\kappa$ is a parameter. This equation has the multisoliton solutions
\begin{equation}\label{nls2}
\mid u(x,t) \mid^2 = (2 \kappa)^{1/2} \frac{d^2}{dx^2} \ln \det \mid \mid B B^* +1 \mid \mid.
\end{equation}

The notation $^*$ means complex conjugation. The matrix $B$ from (\ref{nls2}) is
\begin{equation}\label{nls3}
B_{jk}= \frac{(c_j c_k^*)^{1/2}}{\zeta_j - \zeta_k^*} \exp [i(\zeta_j - \zeta_k^*)x]; \ \ \
c_j(t) = c_j(0) \exp(4 i \zeta^2 t).
\end{equation}

In (\ref{nls3}) $ \zeta = \frac{\lambda p}{1-p^2}$, where $p$ is determined from $\kappa = \frac{2}{1-p^2}$,
$\zeta_i$ are the eigenvalues of problem (8) from \cite{zs} and $\lambda$ is the eigenvalue from (7) of
\cite{zs}, where $\hat L$ is operator (5) from~\cite{zs}.
\par
The transformation $u = u(\phi)$ transforms the nonlinear Schr{\"o}dinger equation to
\begin{equation}\label{nls4}
i \frac{du}{d \phi} \frac{\partial \phi}{\partial t} + \frac{d^2 u}{d \phi^2} \left( \frac{\partial \phi}{\partial x} \right)^2
+ \frac{du}{d \phi} \frac{\partial^2 \phi}{\partial x^2} + \kappa \mid u^2 \mid u = 0.
\end{equation}

Let us specify the form of the transformation to $u = \phi^{\alpha}$, where $\alpha$ is a constant and $\alpha \ne 1$.
Equation (\ref{nls4}) is reduced to
\begin{equation}\label{nls5}
i \alpha \phi \frac{\partial \phi}{\partial t} + \alpha(\alpha-1) \left( \frac{\partial \phi}{\partial x} \right)^2
+ \alpha \phi \frac{\partial^2 \phi}{\partial x^2} + \kappa \phi^2 \mid \phi \mid^{2\alpha} = 0.
\end{equation}

The multisoliton solution of (\ref{nls5}) is $\phi = u^{1/\alpha}$, where $u$ is given by (\ref{nls2}).
\par 
Another possible direction for the extension of the reported research is to apply more complicated
transformations and this will lead to exact solutions of even more complicated
nonlinear differential equations. Moreover, the stability of the obtained solutions can be studied
as in \cite{zs}. The results of such kinds of research will be reported elsewhere.

%

%

\begin{appendix}
\section{Linear Differential Equations and Their Solutions Used in
the Main~Text}\label{appa}
In the main text, we discuss the following linear differential equations and their~solutions.

1. The hyperbolic equation
\begin{equation}\label{wave}
c^2 \frac{\partial^2 u}{\partial x^2} - \frac{\partial^2 u}{\partial t^2} =0,
\end{equation}
where $-\infty \le a <x < b \le + \infty$, and $t>0$.
This is the (1+1)-D wave equation. We consider the general solution
\begin{equation}\label{wavesol}
u(x,t) = F(x+ct) + G(x-ct),
\end{equation}
where $F$ and $G$ are arbitrary $C^2$-functions. In addition, we consider the Cauchy problem
\begin{eqnarray}\label{wave_cauchy}
u(x,0)= f(x), \ \ \ \frac{\partial u}{\partial t}(x,0) = g(x), \ \ 
-\infty < x < \infty, t>0.
\end{eqnarray}

The solution of this problem is given by d'Alembert's formula
\begin{equation}\label{wavesol2}
u(x,t) = \frac{f(x+ct) + f(x-ct)}{2} + \frac{1}{2c} \int \limits_{x-ct}^{x+ct} ds \ g(s).
\end{equation}

Finally, we will use the solution for the case $ 0<x<L$ and initial and
boundary conditions 
\begin{eqnarray}\label{wave_fourierxx}
u(x,0) = f(x), \ \ \ \frac{\partial u}{\partial t}(x,0)=g(x), \ \ 
0 \le x \le L, \nonumber \\
\frac{\partial u}{\partial x}(0,t) = \frac{\partial u}{\partial x}(L,t) = 0, \ \ t \ge 0.
\end{eqnarray}

The solution in this case is 
\begin{eqnarray}\label{wavesol3}
u(x,t) = \sum \limits_{n=1}^\infty \Bigg[a_n \cos \Bigg( \frac{n\pi ct}{L} \Bigg)+b_n \sin \Bigg(\frac{n\pi ct}{L} \Bigg) \Bigg] \sin \Bigg( \frac{n\pi x}{L}\Bigg), \nonumber \\
a_n = \frac{2}{L}\int \limits_0^L dx \  f(x) \sin \Bigg( \frac{n \pi x}{L} \Bigg), \ \ \ b_n = \frac{2}{n \pi c} \int \limits_0^L dx \ 
g(x) \sin \Bigg( \frac{n \pi x}{L} \Bigg).
\end{eqnarray}

2. The parabolic equation
\begin{equation}\label{lhe}
a^2 \frac{\partial^2 u}{\partial x^2} - \frac{\partial u}{\partial t} =0.
\end{equation}

This is the (1 + 1)-D heat equation. 
For the case $-\infty <x < \infty$ and for the initial condition
$u(x,0)=f(x)$, solution is given by the integral of Poisson
\begin{equation}\label{sol_heat1}
u(x,t) = \frac{1}{2a \sqrt{\pi t}} \int \limits_{-\infty}^{+\infty}
d \xi \ f(\xi) \exp \left[ - \frac{(\xi - x)^2}{4 a t^2}\right].
\end{equation}

For the case $0<x<L$ and initial and boundary conditions $u(x,0)=f(x)$
and $u(0,t)=A$, $u(L,t)=B$, the solution is
\begin{eqnarray}\label{sol_heat2}
u(x,t) = A + \frac{B-A}{L} + \sum \limits_{n=1}^{\infty} a_n \exp
\Bigg( - \frac{\pi^2 a^2 n^2}{L^2} t\Bigg ) \sin \Bigg(\frac{n \pi x}{L} \Bigg), \nonumber \\
a_n = - \frac{2}{n \pi}[A+(-1)^{n+1}B] + \frac{2}{L} \int \limits_0^L dx  \ f(x) \sin \Bigg( \frac{n \pi x}{L} \Bigg).
\end{eqnarray}

3. The elliptic equation
\begin{equation}\label{lapl}
\frac{\partial^2 u}{\partial x^2} + \frac{\partial^2 u}{\partial y^2} =0.
\end{equation}

This is the 2D Laplace equation. Here, we consider only the solution for
the case of the rectangle domain $a<x<b$, $c<x<d$ and boundary conditions
$u(a,y)=f(y)$; \mbox{$u(b,y)=g(y)$}; $u(x,c)=h(c)$; $u(x,d) = k (x)$.
The solution is
\begin{equation}\label{sol_lapl}
u = u_1 + u_2; \ \ \ u_1 = \sum \limits_n X_n(x) Y_n(y); \ \ \ u_2 = \sum \limits_m Z_m(x) V_m(y),
\end{equation}
where the boundary conditions for $u_{1,2}$ are
$u_1 (a, y) = f (y), u_1 (b, y) = g(y), u_1 (x, c) = 0, u-1 (x, d) = 0$;
$u_2 (a, y) = 0, u_2 (b, y) = 0, u_2 (x, c) = h(x), u_2 (x, d) = k(x)$,
where $X_n,Y_n,Z_m,V_m$ are solutions of the equations
\begin{eqnarray}\label{lapl2}
\frac{d^2X}{dx} - \lambda_1 X(x) = 0 , \ a<x<b, \nonumber \\
\frac{d^2Y}{dy} - \lambda_2 Y(y) = 0 , \ c<y<d, \nonumber \\
\frac{d^2Z}{dx} - \lambda_3 Z(x) = 0 , \ a<x<b,  \\
\frac{d^2V}{dy} - \lambda_4 V(y) = 0 , \ c<y<d.\nonumber
\end{eqnarray}

\end{appendix}


\begin{thebibliography}{999}
\bibitem{m1}
Latora, V.; Nicosia, V.; Russo, G. {\em Complex Networks. Principles, Methods, and Applications};
Cambridge University Press: Cambridge, UK, 2017; ISBN 978-1-107-10318-4.
\bibitem{m2}
Vitanov, N.K. {\em Science Dynamics and Research Production. Indicators, Indexes, Statistical Laws and 
Mathematical Models}; Springer: Cham, Switzerland, 2016; ISBN 978-3-319-41629-8.
\bibitem{m3}
Treiber, M.; Kesting, A. {\em Traffic Flow Dynamics: Data, Models, and
 Simulation}; Springer: Berlin/Heidelberg, Germany, 2013; ISBN~978-3-642-32460-4. 
\bibitem{m4}
Dimitrova, Z.I. Flows of Substances in Networks and Network 
Channels: Selected Results and Applications. {\em Entropy} {\bf 
2022}, {\em 24}, 1485. {https://doi.org/10.3390/e24101485}.
\bibitem{m5}
Drazin, P.G. {\em Nonlinear Systems}; Cambridge University Press: Cambridge, UK, 1992;
ISBN~0-521-40489-4.
\bibitem{m6}
Kantz, H.; Schreiber, T. {\em Nonlinear Time Series Analysis}; Cambridge University Press: 
Cambridge, UK, 2004; ISBN~978-0511755798.
\bibitem{m7}
Verhulst, F. {\em Nonlinear Differential Equations and Dynamical Systems}; Springer: Berlin/Heidelberg, Germany, 2006; ISBN~978-3-540-60934-6.
\bibitem{m7a}
Popivanov, P.; Slavova, A. {\em Nonlinear Waves: An Introduction}; World Scientific: Singapore, 2010; ISBN~9789813107953.
\bibitem{m7b}
Debnath, L. (Eds.) {\em Nonlinear Waves}; Cambridge University Press: 
Cambridge, UK, 1983; ISBN~0-521-25468-X.
\bibitem{m7c}
Kulikovskiii, A.; Sveshnikova, E. {\em Nonlinear Waves in Elastic Media}; CRC 
Press: Boca Raton, FL, USA, 2021; ISBN~0-8493-8643-8.
\bibitem{m7d}
Ma, Q. {\em Advances in Numerical Simulation of Nonlinear Water Waves}; World Scientific: Singapore, 2010; ISBN~ 9789812836502.
\bibitem{m7e}
Osborne, A.R. {\em Nonlinear Topics in Ocean Physics}; North-Holland: Amsterdam, The Netherlands, 1991; ISBN~9780444597823.
\bibitem{m7f}
Nazarov, V.; Radostin, A. {\em Nonlinear Acoustic Waves in Micro-Inhomogeneous 
Solids}; Wiley: Chchester, UK,  2005; ISBN~9781118456088.
\bibitem{m7g}
Kim, C.-H. {\em Nonlinear Waves and Offshore Structures}; World Scientific: Singapore, 2008; ISBN~ 9789813102484.  
\bibitem{m7h}
Fillipov, A.T. {\em The Versatile Soliton}; Springer: New York, 2010; ISBN~9780817649746.
\bibitem{m7i}
Akhmediev, N.; Ankiewicz, A. {\em Dissipative Solitons}; Springer: Berlin, Germany,  2005, ISBN~9783540233732.
\bibitem{m7j}
Davydov, A.S.  {\em Solitons in Molecular Systems}; Springer: Dordrecht, The Netherlands, 2013; ISBN~ 9789401730259.
\bibitem{m7k}
Olver, P.J.; Sattiger, D.H. {\em Solitons in Physics, Mathematics, and Nonlinear Optics}; Springer: New York, NY, USA, 2012; ISBN~9781461390336.
\bibitem{m8}
Ablowitz, M.J.; Clarkson, P.A. {\em Solitons, Nonlinear Evolution Equations and Inverse Scattering}; 
Cambridge University Press: Cambridge, UK, 1991; ISBN~978-0511623998.
\bibitem{m8a}
Ablowitz, M. J.; Kaup, D. J.; Newell, A. C.; Segur, H. The Inverse Scattering 
Transform‐Fourier Analysis for Nonlinear Problems. {\em Studies in Applied 
Mathematics}, {\bf 1974}, {\em 53} (4), 249 -- 315.
{https://doi.org/10.1002/sapm1974534249}
 \bibitem{m8b}
Ablowitz, M. J.; Musslimani, Z. H. Inverse Scattering Transform for the 
Integrable Nonlocal Nonlinear Schr{\"o}dinger Equation. {\em Nonlinearity}, {\bf 2016}, {\em 29} (3), 915. {https://doi.org/10.1088/0951-7715/29/3/915}
\bibitem{m8c}
Vitanov, N. K. Simple Equations Method (SEsM) and its Connection with the 
Inverse Scattering Transform Method. {\em AIP Conference Proceedings} {\bf 2021}, {\em 2321}, 030035,  {https://doi.org/10.1063/5.0040409}
\bibitem{m8d}
Fokas, A. S.; Ablowitz, M. J. The Inverse Scattering Transform for the Benjamin‐Ono Equation—A Pivot to Multidimensional Problems. {\em Studies in Applied Mathematics}, {\bf 1983}, {\em 68} (1), 1-10.
{https://doi.org/10.1002/sapm19836811}
 \bibitem{m8e}
 Zhang, X.; Chen, Y. Inverse Scattering Transformation for Generalized Nonlinear Schr{\"o}dinger Equation. {\em Applied Mathematics Letters}, {\bf 2019}, {\em 98}, 306 -- 313. {https://doi.org/10.1016/j.aml.2019.06.014}
\bibitem{m8f}
Osborne, A. R. The Inverse Scattering Transform: Tools for the Nonlinear 
Fourier Analysis and Filtering of Ocean Surface Waves. {\em Chaos, Solitons \& 
Fractals}, {\bf 1995} {\em 5} (12), 2623 -- 2637.
{https://doi.org/10.1016/0960-0779(94)E0118-9}
\bibitem{m8g}
Osborne, A. R. Soliton Physics and the Periodic Inverse Scattering Transform. 
{\em Physica D},{\bf 1995}, {\em 86} (1-2), 81 -- 89.
{https://doi.org/10.1016/0167-2789(95)00089-M}
\bibitem{m8h}
Ji, J. L.; Zhu, Z. N. Soliton Solutions of an Integrable Nonlocal Modified 
Korteweg–de Vries Equation Through Inverse Scattering Transform. {\em Journal 
of Mathematical Analysis and Applications}, {\bf 2017}, 453(2), 973 -- 984.
{https://doi.org/10.1016/j.jmaa.2017.04.042}
\bibitem{m9}
Hirota, R. {\em The Direct Method in Soliton Theory}; Cambridge University Press: Cambridge, UK, 2004; ISBN 0-521-83660-3.
\bibitem{m10}
Gibbon, J. D.; Radmore, P.; Tabor, M.; Wood, D. The Painleve Property and Hirota's Method. {\em Studies in Applied Mathematics}, {\bf 1985}, {\em 72} (1), 39 -- 63.  {https://doi.org/10.1002/sapm198572139}
\bibitem{m11}
G{\"u}rses, M.; Pekcan, A. Nonlocal Modified KdV Equations and Their Soliton 
Solutions by Hirota Method. {\em Communications in Nonlinear Science and 
Numerical Simulation}, {\bf 2019}, 67, 427-448. {https://doi.org/10.1016/j.cnsns.2018.07.013}
\bibitem{m12}
Zhou, Y.; Ma, W. X. Complexiton Solutions to Soliton Equations by the Hirota 
Method. {\em Journal of Mathematical Physics}, {\bf 2017}, {\em 58} (10).
{https://doi.org/10.1063/1.4996358}
\bibitem{m13}
Jia, T. T.; Chai, Y. Z.; Hao, H. Q. Multi-soliton Solutions and Breathers for 
the Generalized Coupled Nonlinear Hirota Equations via the Hirota Method. 
{\em Superlattices and Microstructures}, {\bf 2017}, {\em 105}, 172-182.
{https://doi.org/10.1016/j.spmi.2016.10.091}
\bibitem{m14}
Ma, W. X. Soliton Solutions by Means of Hirota Bilinear Forms. {\em Partial 
Differential Equations in Applied Mathematics}, {\bf 2022}, 5, 100220.
{https://doi.org/10.1016/j.padiff.2021.100220}
\bibitem{l1}
Infeld, E.; Rowlands, G. {\em Nonlinear Waves, Solitons and Chaos}; Cambridge
University Press: Cambridge, UK, 2000; ISBN~0-521-63212-9.
\bibitem{l2}
Zhao, Z.-L.; He, L.-C.; Wazwaz, A.-M.  Dynamics of Lump Chains for the BKP Equation
Describing Propagation of Nonlinear Wave. {\em Chin. Phys. B} {\bf 2023},  {\em 32}, 040501.
{https://doi.org/10.1088/1674-1056/acb0c1}.
\bibitem{l3}
Ablowitz, M.J. {\em Nonlinear Dispersive Waves: Asymptotic Analysis and Solitons};
Cambridge University Press: Cambridge, UK, 2011; ISBN~9781107012547.
\bibitem{t1}
Calogero, F.; Degasperis, A. {\em Spectral Transform and Solitons}; North Holland:
Amsterdam, The Netherlands, 1982; ISBN 0-444-86368-0.
\bibitem{t2}
Osborne, A.R.; Bergamasco, L.  The Solitons of Zabusky and Kruskal revisited: Perspective in
Terms of the Periodic Spectral Transform. {\em Physica  D} {\bf 1986}, {\em 18}, 26--46.
{https://doi.org/10.1016/0167-2789(86)90160-0}.
\bibitem{t3}
Wadati, M.; Sogo, K. Gauge Transformations in Soliton theory. {\em J. Phys. Soc. Jpn.}
{\bf 1983}, {\em 52}, 394--398. {https://doi.org/10.1143/JPSJ.52.394}.
\bibitem{t3a}
Buccoliero, D.; Desyatnikov, A.S. Quasi-periodic Transformations of Nonlocal Spatial Solitons.
{\em Opt. Express} {\bf 2009}, \emph{17}, 9608--9613. {https://doi.org/10.1364/OE.17.009608}.
\bibitem{t4}
Date, E.; Jimbo, M.; Kashiwara, M.; Miwa, T. Transformation Groups for Soliton Equations: IV.
A New Hierarchy of Soliton Equations of KP-type.  {\em Physica  D} {\bf 1982}, {\em 4}, 343--365. {https://doi.org/10.1016/0167-2789(82)90041-0}.
\bibitem{l4}
Whitham, G.B.;  {\em Linear and Nonlinear Waves}; Wiley: New York, NY, USA, 1999; ISBN~0-471-35942-4.
\bibitem{l4a}
Zakharov, V.E.;  Wabnitz, S. (Eds.) {\em Optical Solitons: Theoretical Challenges and
Industrial Perspectives}; Springer: Berlin/Heidelberg, Germany, 2013; ISBN~9783662038079.
\bibitem{l4b}
Gibbon, J.D. A Survey of the Origins and Physical Importance of Soliton Equations.
{\em Philos. Trans. R. Soc. Lond. Ser. A Math. Phys. Sci.} {\bf 1985},  {\em 315},  335--365.
\bibitem{4b1}
Ur Rehman, H.; Awan, A.U.; Habib, A.; Gamaoun, F.;  El Din, M.T.; Galal, A.M.
Solitary Wave Solutions for a Strain Wave Equation in a microstructured Solid.
{\em Results Phys.} {\bf 2022}, {\em 39}, 105755. {https://doi.org/10.1016/j.rinp.2022.105755}.
\bibitem{l4c}
Newell, A.C. The General Structure of Integrable Evolution Equations. {\em Proc. R. Soc. Lond. A Math. Phys. Sci.} {\bf 1979}, {\em 365},  283--311. {https://doi.org/10.1098/rspa.1979.0018}.
\bibitem{l4d}
Yang, J. {\em Nonlinear Waves in Integrable and Nonintegrable Systems}; SIAM: Philadelphia, PA, USA, 2010; ISBN~9780898719680.
\bibitem{l5}
Rogers, C.; Schief, W.K. {\em B{\"a}cklund and Darboux Transformations: Geometry
and  Modern Applications in Soliton Theory}; Cambridge University Press: Cambridge,  UK,
2002; ISBN~9780521012881.
\bibitem{se1}
Vitanov, N.K. Simple Equations Method (SEsM): An Effective Algorithm 
for Obtaining Exact Solutions of Nonlinear Differential Equations. 
{\em Entropy} \textbf{2022}, {\em 24}, 1653. {https://doi.org/10.3390/e24111653}.
\bibitem{se2}
Vitanov, N.K.; Dimitrova, Z.I.; Vitanov, K.N. Simple Equations Method (SEsM): Algorithm, Connection
with Hirota Method, Inverse Scattering Transform Method, and Several Other Methods. {\em Entropy}
{\bf 2021}, {\em 23}, 10. {https://doi.org/10.3390/e23010010}.
\bibitem{se3}
Vitanov, N.K.; Dimitrova, Z.I.  Simple Equations Method and Non-linear Differential Equations 
with Non-polynomial Non-linearity. {\em Entropy} {\bf 2021}, \emph{23}, 1624. {https://doi.org/10.3390/e23121624}.
\bibitem{se4}
Vitanov, N.K.; Dimitrova, Z.I.; Vitanov, K.N. On the Use of Composite Functions in the Simple 
Equations Method to Obtain Exact Solutions of Nonlinear Differential Equations. {\em Computation} 
{\bf 2021}, \emph{9}, 104. {https://doi.org/10.3390/computation9100104}.
\bibitem{se5}
Vitanov, N. K.; Dimitrova, Z.I. Modified Method of Simplest equation Applied to the Nonlinear Schr{\"o}dinger Equation. {\em Journal of Theoretical and Applied Mechanics}, {\bf 2018}, {\em 48} (1), 59 -- 69.
{https://doi.org/10.2478/jtam-2018-0005}
\bibitem{se6}
Vitanov, N. K. Modified Method of Simplest Equation: Powerful tool for 
Obtaining Exact and approximate Traveling-wave solutions of Nonlinear PDEs. 
{\em Communications in Nonlinear Science and Numerical Simulation}, {\bf 2011},
{\em 16} (3), 1176 -- 1185. {https://doi.org/10.1016/j.cnsns.2010.06.011}
\bibitem{se7}
Vitanov, N. K.; Dimitrova, Z. I.; Kantz, H. (2010). Modified Method of Simplest 
equation and its Application to Nonlinear PDEs. {\em Applied Mathematics and 
Computation}, {\bf 2010}, {\em 216} (9), 2587 -- 2595.
{https://doi.org/10.1016/j.amc.2010.03.102}
\bibitem{se8}
Vitanov, N. K. On Modified Method of Simplest Equation for Obtaining Exact and 
Approximate Solutions of Nonlinear PDEs: The role of the Simplest Equation. 
{\em Communications in Nonlinear Science and Numerical Simulation}, {\bf 2011}, {\em 16} (11), 4215 -- 4231. {https://doi.org/10.1016/j.cnsns.2011.03.035}
\bibitem{se9}
Vitanov, N. K.; Dimitrova, Z. I.; Vitanov, K. N. Modified Method of Simplest 
Equation for Obtaining Exact Analytical Solutions of Nonlinear Partial 
Differential Equations: Further development of the Methodology with 
Applications. {\em Applied Mathematics and Computation}, {\bf 2015}, {\em 269}, 
363 -- 378. {https://doi.org/10.1016/j.amc.2015.07.060}
\bibitem{hopf}
Hopf, E. The Partial Differential Equation: $u_t + uu_x= \epsilon u_{xx}$. {\em Commun. Pure Appl. Math.} {\bf 1950}, {\em 3},  201--230. {https://doi.org/10.1002/cpa.3160030302}.
\bibitem{cole}
Cole, J.D. On a Quasi-Linear Parabolic Equation Occurring in Aerodynamics. {\em Q. Appl. Math.} {\bf 1951}, {\em 9}, 225--236. {https://doi.org/10.1090/qam/42889}.
\bibitem{bk1}
Wahlquist, H.D.; Estabrook, F.B. B{\"a}cklund Transformation for Solutions of the
Korteweg-de Vries Equation. {\em Phys. Rev. Lett.} {\bf 1973} {\em 31}, 1386--1390. {https://doi.org/10.1103/PhysRevLett.31.1386}.
\bibitem{bk2}
Dodd, R.K.; Bullough, R.K. B{\"a}cklund Transformations for the Sine–Gordon Equations.
{\em Proc. R. Soc. Lond. A} {\bf 1976}, {\em 351}, 499--523. {https://doi.org/10.1098/rspa.1976.0154}.
\bibitem{bk3}
Satsuma, J.; Kaup, D.J. A B{\"a}cklund Transformation for a Higher Order Korteweg-de Vries Equation.
{\em J. Phys. Soc. Jpn.} {\bf 1977}, \emph{43}, 692--697. {https://doi.org/10.1143/JPSJ.43.692}.
\bibitem{bk3a}
Hirota, R.; Satsuma, J. A Variety of Nonlinear Network Equations Generated from the B{\"a}cklund
Transformation for the Toda Lattice. {\em Prog. Theor. Phys. Suppl.} {\bf 1976}, {\em 59}, 64--100. {https://doi.org/10.1143/PTPS.59.64}.
\bibitem{bk4}
Gao, X.Y.; Guo, Y.J.;  Shan, W.R. Regarding the Shallow Water in an Ocean via a Whitham-Broer-Kaup-like System:
Hetero-Bäcklund Transformations, Bilinear Forms and M Solitons. {\em Chaos Solitons Fractals}
{\bf 2022}, {\em 162}, 112486. {https://doi.org/10.1016/j.chaos.2022.112486}.
\bibitem{bk4a}
Hirota, R. A New Form of B{\"a}cklund Transformations and its Relation to the Inverse Scattering Problem. {\em Progress of Theoretical Physics}, {\bf 1974}, {\em 52}(5), 1498 -- 1512. {https://doi.org/10.1143/PTP.52.1498}
\bibitem{bk4b}
Lamb Jr, G. L.  B{\"a}cklund Transformations for Certain Nonlinear Evolution 
Equations. {\em Journal of Mathematical Physics}, {\bf 1974}, {\em 15} (12), 
2157 -- 2165. {https://doi.org/10.1063/1.1666595}
\bibitem{bk4c}
Fan, E. Auto-B{\"a}cklund Transformation and Similarity Reductions for General 
Variable Coefficient KdV Equations. {\em Physics Letters A},{\bf 2002}, {\em 294} (1), 26 -- 30. {https://doi.org/10.1016/S0375-9601(02)00033-6}
\bibitem{bk4d}
Zhou, T. Y.; Tian, B.; Chen, S. S.; Wei, C. C.; Chen, Y. Q. B{\"a}cklund 
Transformations, Lax pair and Solutions of a Sharma-Tasso-Olver-Burgers 
Equation for the Nonlinear Dispersive Waves. {\em Modern Physics Letters B}, {\bf 2021}, {\em 35} (35), 2150421. {https://doi.org/10.1142/S0217984921504212}
\bibitem{bk4e}
Wang, G.; Liu, Q. P.; Mao, H. The Modified Camassa–Holm eEquation: B{\"a}cklund 
Transformation and Nonlinear Superposition Formula. {\em Journal of Physics A: Mathematical and Theoretical}, {\bf 2020}, {\em 53} (29), 294003.
{https://doi.org/10.1088/1751-8121/ab7136}
\bibitem{bk4f}
Gao, X. Y.; Guo, Y. J.;  Shan, W. R. Bilinear Forms through the Binary Bell 
Polynomials, N Solitons and B{\"a}cklund Transformations of the Boussinesq–
Burgers System for the Shallow Water Waves in a Lake or Near an Ocean Beach. 
{\em Communications in Theoretical Physics}, {\bf 2020}, {\em 72} (9), 095002.
{https://doi.org/10.1088/1572-9494/aba23d}
\bibitem{bk5}
Wang, K.J. B{\"a}cklund Transformation and Diverse Exact Explicit Solutions
of the Fractal Combined Kdv--mkdv Equation. {\em  Fractals} {\bf 2022}, {\em 30}, 2250189.
{https://doi.org/10.1142/S0218348X22501894}.
\bibitem{db1}
Guo, B.; Ling, L.; Liu, Q.P. Nonlinear Schr{\"o}dinger Equation: Generalized Darboux Transformation
and Rogue Wave Solutions. {\em Phys. Rev. E} {\bf 2012}, {\bf 85}, 026607. {https://doi.org/10.1103/PhysRevE.85.026607}.
\bibitem{db2}
Xu, S.; He, J.;  Wang, L. The Darboux Transformation of the Derivative Nonlinear Schr{\"o}dinger Equation.
{\em J. Phys. A Math. Theor.} {\bf 2011}, {\em 44}, 305203. {https://doi.org/10.1088/1751-8113/44/30/305203}.
\bibitem{db3}
Gu, C.; Hu, H.; Zhou, Z. {\em Darboux Transformations in Integrable 
Systems: Theory and their Applications to Geometry}. Springer: Dordrecht, 2005;
ISBN 1-4020-3087-8
\bibitem{db4}
Gu, C.; Hu, H.; Zhou, Z. {\em  Darboux Transformations in Integrable
Systems: Theory and Their Applications to Geometry}; Springer Science \& Business Media:  Dordrecht, The Netherlands, 2005; ISBN~1-4020-3087-8.
\bibitem{db5}
Li, Y.; Zhang, J. E. Darboux Transformations of Classical Boussinesq system and 
its Multi-soliton Solutions. {\em Physics Letters A}, {\bf 2001}, {\em 284} (6), 253 -- 258. {https://doi.org/10.1016/S0375-9601(01)00331-0}
\bibitem{db6}
Ma, W. X.; Zhang, Y. J. Darboux Transformations of Integrable Couplings and 
Applications. {\em Reviews in Mathematical Physics}, {\bf 2018}, {\em 30}(02), 
1850003. {https://doi.org/10.1142/S0129055X18500034}
\bibitem{db7}
Aktosun, T.; Van der Mee, C.  A Unified Approach to Darboux 
Transformations. {\em Inverse Problems}, {\bf 2009}, {\em 25} (10), 105003.
{https://doi.org/10.1088/0266-5611/25/10/105003}
\bibitem{db8}
Xu, S.; He, J.; Wang, L.  The Darboux Transformation of the Derivative 
Nonlinear Schr{\"o}dinger equation. {\em Journal of Physics A: Mathematical and 
Theoretical}, {\bf 2011}, {\em 44} (30), 305203.
{https://doi.org/10.1088/1751-8113/44/30/305203}
\bibitem{db9}
Qiu, D.; He, J.; Zhang, Y.; Porsezian, K. The Darboux Transformation of the 
Kundu–Eckhaus Equation. {\em Proceedings of the Royal Society A}, {\bf 2015}, {\em 471} (2180), 20150236. {https://doi.org/10.1098/rspa.2015.0236}
\bibitem{db10}
Matveev, V. B. Darboux Transformation and Explicit solutions of the Kadomtcev-Petviaschvily equation, depending on Functional Parameters. {\em Letters in Mathematical Physics}, {\bf 1979}, {\em 3}, 213 -- 216.
{https://doi.org/10.1007/BF00405295}
\bibitem{zs}
Zakharov, V.E.; Shabat, A.V. Exact Theory of Two-dimensional Self-Focusing
and One-dimensional Self-modulation of Waves in Nonlinear Media. {\em Sov. Phys. JETP}
{\bf 1972}, {\em 34}, 62--69.





\end{thebibliography}
\end{document}